\documentclass[oneside,reqno]{amsart}
\usepackage{amsmath, amsthm,amssymb,amsfonts,latexsym}
\usepackage[dvips]{graphicx}
\usepackage{color}
\usepackage{array}
\usepackage[hmargin=2.75cm,vmargin=2.75cm]{geometry}
\usepackage{booktabs, multicol, multirow}
\usepackage{verbatim}
\usepackage{natbib}
\newtheorem{prop}{Proposition}

\newtheorem{lemma}{Lemma}
\newtheorem{thm}{Theorem}

\usepackage{tikz}
\usepackage{caption}
\usepackage{subcaption}
\usetikzlibrary{trees,shapes,snakes}

\theoremstyle{remark}
\newtheorem{remark}{Remark}

\allowdisplaybreaks

\title{Dividend Policy and Capital Structure of a Defaultable Firm}
\thanks{The author would like to thank David Hobson, Bart Lambrecht, Harry Zhang and the seminar participants at Imperial College London for the useful discussions and comments.}
\author{Alex S.L. Tse}
\thanks{
	Department of Mathematics, Imperial College London, London SW7 2AZ, UK. a.tse@imperial.ac.uk}
\date{\today}
\begin{document}

\keywords{Dividend policy, capital structure, default risk, singular stochastic control, HJB equation.}

\maketitle

\begin{abstract}
	Default risk significantly affects the corporate policies of a firm. We develop a model in which a limited liability entity subject to Poisson default shock jointly sets its dividend policy and capital structure to maximize the expected lifetime utility from consumption of risk averse equity investors. We give a complete characterization of the solution to the singular stochastic control problem. The optimal policy involves paying dividends to keep the ratio of firm's equity value to investors' wealth below a critical threshold. Dividend payout acts as a precautionary channel to transfer wealth from the firm to investors for mitigation of losses in the event of default. Higher the default risk, more aggressively the firm leverages and pays dividends.
	
\end{abstract}
	
\section{Introduction}

Since the capital structure irrelevance principle and dividend irrelevance principle of \cite{modigliani-miller58} and \cite{miller-modigliani61}, a vast literature has emerged to explore the factors driving corporate policies observed in practice. Historically, some important considerations include tax benefit, asymmetric information, signaling motive, agency costs, financial distress costs, managerial risk aversion and etc. While market frictions and strategic interaction among agents are all realistic concerns, very fundamental factor such as default risk could indeed also play a crucial role behind corporate finance decisions. In this paper, we examine the impact of default risk of a firm on its joint decision of dividend policy and capital structure as well as equity investors' consumption behavior. 

Our model features a limited liability firm and risk averse equity investors. At each point of time given the amount of equity capital in place, the firm simultaneously decides how much to invest in a risky asset (which implies the amount of debt required and in turn its choice of capital structure) and how much to pay out to investors with logarithm utility function. Investors can deposit the dividends received in a riskfree retail saving account and consume to derive utility flow. The interests of firm managers (who set the capital structure decision and payout policy) and investors (who choose their own consumption policy) are perfectly aligned such that their joint economic objective here is to maximize investors' expected lifetime utility of consumption. 

We assume wealth can only flow from the firm to investors in form of dividends but capital cannot be injected into the firm from investors via equity issuance. A reduced form approach is adopted to model firm's default where the equity value of the firm jumps to zero upon the arrival of a Poisson shock. As a result, our model is perhaps the most suitable to describe the corporate policies of small businesses and start-ups which face a high barrier to equity financing and a high level of operational or financial uncertainty (and hence are prone to default). 

The incentive to pay out dividends in presence of default risk is intuitive. A unit of wealth retained in the firm can generate economic value via investment in the risky asset but it comes at the cost of the potential loss when default arrives. Equity investors, however, can exploit the limited liability structure of the entity by transferring some of the firm's equity to their private pocket via dividends to earn a mediocre but safe retail interest rate, rather than leaving all the money on the table just to be wiped out (for example, to be seized by the creditors) when the firm goes bust.

The underlying optimization problem turns out to be a singular stochastic control problem. In contrast to many other models of investment and dividend distribution, the optimal payout policy in our setup is to pay dividends as to keep the ratio of the firm's equity value to investors' wealth level below a certain critical threshold, rather than just to pay out the cash to keep the equity reserve below a constant level. Moreover, the optimal leverage level is state-dependent with its magnitude being decreasing in the equity-to-wealth ratio instead of being a constant. 

Analysis of such a singular stochastic control problem is not straightforward in general and the concept of viscosity solutions sometimes has to be invoked. Nonetheless, based on the transformation techniques introduced in the recent work of \cite{hobson-zhu16} and \cite{hobson-tse-zhu18}, we show that the associated HJB equation can indeed be reduced to a first order crossing problem. The critical dividend payment boundary can be read from the point at which the solution to a first order differential equation first crosses a given analytical function. An important advantage of this approach is that it is relatively easy to deduce the comparative statics of the model. We find that a higher default risk leads to a lower critical equity-to-wealth ratio for dividend payout, a higher leverage level and a lower consumption rate of investors.

We close the introduction by relating our work to the existing literature. Both payout policy and capital structure decision are long-standing research topics in corporate finance and it is impossible to give a full account of the theory development here. Instead, we refer readers to the excellent surveys by \cite{harris-raviv91} and \cite{allen-michaely03}. Mainstream finance literature often studies payout policy and capital structure decision separately but not their joint interaction, as highlighted by \cite{lambrecht-myers12}. The economic foundation of our model is based on \cite{lambrecht-myers17} and \cite{lambrecht-tse18} where investment and payout policy are examined together in an inter-temporal investment/consumption model. Although our exposition assumes perfect coordination between managers and investors, the same modeling framework can be directly applied to an agency setup driven by utility maximization of self-interested risk averse managers (see Remark \ref{remark:agency} in Section \ref{sect:setup}).

Continuous-time portfolio selection, which is the skeleton of the modeling framework in this paper, is of course another enormous field in the mathematical finance and stochastic control literature. Existing work usually focuses on optimal consumption/investment models as per the seminal work of Merton (\citeyear{Merton:69}, \citeyear{Merton:71}) and its many other variants, or optimal risk control/dividend distribution models which are studied extensively in the field of insurance [e.g. \cite{radner-shepp96}, \cite{browne97} and \cite{jgaard-taksar99}]. In a certain sense, these two classes of models are equivalent because consumption and dividend are usually viewed interchangeably, as stated by \cite{taksar00} that optimal risk control/dividend models in most instances are consumption/investment models with linear utility function and arithmetic Brownian motion return. Implicitly, these models assume that payout from a firm has to be consumed immediately and therefore discard the possibility that dividends can at least be deposited for consumption later. We disentangle the effect of dividend payout and consumption by introducing a riskfree retail saving account to investors as an outside option. One novel prediction as a result of this flexibility is that a bad firm will be voluntarily liquidated because there is no longer the need to keep a bad firm alive just for the purpose of generating dividends over time to match the smooth consumption required by risk averse investors. To the best of our knowledge, our current paper is the first one to consider a joint, dynamic model of capital structure (i.e. investment/risk control), dividend payout and individuals' consumption. 

We also examine the joint impact of default risk on the firms' corporate policies and consumption pattern of investors. Incorporation of exogenous default risk is not a new mathematical feature - optimal portfolio choice and consumption problems with random termination time have been considered in life insurance models [see for example \cite{richard75} and \cite{pliska-ye07}]. In our current context of corporate finance, nonetheless, consideration of default risk leads to some interesting economic phenomena as revealed by the comparative statics.\footnote{The interaction among dividend, leverage and firm's default is also explored in a one-period signaling model of \cite{kucinskas18}, where high dividend is a bad signal for firms with high leverage because the payout can be driven by the motivation of ``cash out'' prior to bankruptcy.}

One important aspect of our model is that equity financing is not possible. Thus our model is somewhat similar to a Merton problem with transaction costs as in \cite{magill-constantinides76}, \cite{davis-norman90} and \cite{shreve-soner94}. More precisely, the special case studied by \cite{hobson-zhu16}, where transaction cost is zero on sale and infinite on purchase, is comparable to our model in which equity capital can only be passed to investors as dividends but fresh capital cannot be injected into the firm. Consequently, our optimal dividend strategy is similar to the investment strategy obtained by \cite{hobson-zhu16}. However, our model is inherently a higher dimensional one because of the leverage decision involved and thus is not a trivial extension of their model. See the discussion in Section \ref{sect:setup}. Broadly speaking, our work contributes to the growing literature on solving a singular stochastic control problem via reduction to a first order crossing problem [e.g. \cite{hobson-zhu16}, \cite{hobson-tse-zhu18} and \cite{hobson-tse-zhu16}]. It showcases the mathematical techniques are amendable outside the context of portfolio selection under transaction costs and how powerful comparative statics can be derived based on simple comparison principles. 

The rest of the paper is organized as follows. Section \ref{sect:setup} introduces the modeling setup. Section \ref{sect:mainresult} gives the main results of the paper and their economic significance. A heuristic derivation of the solution is provided in Section \ref{sect:heuristics} and we give a full verification argument of the candidate solution in Section \ref{sect:veri}. Section \ref{sect:conclude} concludes. Some proofs in the main body of the paper are deferred to the appendix.

\section{The setup}
\label{sect:setup}
	
	Throughout this paper we work with $(\Omega,\mathcal{F},\{\mathcal{F}_t\},\mathbb{P})$ a filtered probability space satisfying the usual conditions which supports a one-dimensional Brownian motion $B=(B_t)_{t\geq 0}$.
	
	A firm can invest in two classes of asset: a bond instrument with interest rate $\rho$, and a risky asset which price process is a geometric Brownian motion with drift $\mu$ and volatility $\sigma>0$. For every unit of equity within the firm at time $t$, an amount of $\pi_t$ is invested in the risky asset whereas $1-\pi_t$ is invested in the bond. A choice of $\pi_t>1$ corresponds to a levered firm which borrows an amount of $\pi_t-1$ to finance its investment in the risky asset. We call $\pi=(\pi_t)_{t\geq 0}$ an investment policy of the firm which is required to be adapted and satisfy $\int_{0}^t \pi^2_{u}(\omega) du<\infty$ for all $(t,\omega)$. Equity within the firm can also be distributed to risk averse equity investors in form of dividends. Let $\Phi=(\Phi_t)_{t\geq 0}$ with $\Phi_{0-}=0$ be an adapted, non-decreasing process representing the cumulative amount of dividends paid to the investors up to time $t$. The equity value of the firm $S=(S_t)_{t\geq 0}$ then evolves as
	\begin{align}
	dS_t&=\pi_t S_t(\mu dt + \sigma dB_t)+(1-\pi_t) S_t \rho dt - d\Phi_t \nonumber \\
	&=[\rho+(\mu-\rho)\pi_t]S_tdt+\sigma\pi_t S_t dB_t-d\Phi_t.
	\label{eq:EquityValSDE}
	\end{align}
	
	The risk averse investors have a logarithm utility function. They possess a private account which earns a retail riskfree rate of $r$ and they consume at the same time to derive utility flow continuously. A consumption policy $c=(c_t)_{t\geq 0}$ is a non-negative, adapted process with $\int_0 ^t c_u(\omega) du<\infty$ for all $(t,\omega)$. The investors' wealth level $X=(X_t)_{t\geq 0}$ then follows the dynamic
	\begin{align}
	dX_t=(rX_t-c_t)dt+d\Phi_t.
	\label{eq:PrivateWealthSDE}
	\end{align} 
	
	The firm is exposed to a Poisson shock with intensity $\lambda>0$ which causes the firm to default and wipes out its equity entirely. Equity investors are protected by the limited liability structure of the entity and their private wealth will remain intact. After the firm's default, there is no other investment opportunity available to the investors except their private saving account. Hence their optimal consumption strategy post-default can be derived by solving the deterministic control problem of
	\begin{align}
	F(x):=\sup_{c_t>0}\int_{0}^\infty e^{-\beta t}\ln c_t dt
	\label{eq:DeterControlProb}
	\end{align}
	under the dynamic $dX_t=(rX_t-c_t)dt$ with $X_{0}=x$. Here $\beta>0$ is the investors' subjective discount rate. The solution to \eqref{eq:DeterControlProb} is known as
	\begin{align}
	F(x)=\frac{1}{\beta}\ln x+\frac{1}{\beta}\left[\frac{r}{\beta}+\ln\beta-1\right]
	\end{align}
	and the corresponding optimal consumption strategy is given by $c_t^*=\beta X_t$.
		
	A collection of consumption, investment and dividend policies $(c,\pi, \Phi)$ is said to be admissible if $S_t$ and $X_t$ are non-negative with $(S_t,X_t)\notin (0,0)$ for all $t\geq 0$. Denote by $\mathcal{A}(s,x)$ the class of admissible strategies with initial value $(S_{0-} = s, X_{0-} = x)$. Prior to default, equity investors' expected discounted liftime utility from consumption under a given $(c,\pi, \Phi)$ is
	\begin{align*}
	J(s,x;c,\pi,\Phi)&:=\mathbb{E}\left[\int_0^{\infty}e^{-\beta t}\ln c_t dt\Biggr | S_{0-} = s, X_{0-} = x\right] \\
	&=\mathbb{E}\left[\int_0^{\tau}e^{-\beta t}\ln c_t dt+\int_{\tau}^{\infty}e^{-\beta t}\ln c_t dt\Biggr | S_{0-} = s, X_{0-} = x\right]
	\end{align*}
	where $\tau$ is an exponential random variable with parameter $\lambda$ defined on the same probability space and it is independent of the underlying Brownian motion $B$. Firm managers act in the best interest of the investors. Their joint objective is to solve
	\begin{align*}
	V(s,x):=\sup_{(c,\pi, \Phi)\in \mathcal{A}(s,x)}J(s,x;c,\pi,\Phi)
	\end{align*}
	which can be rewritten as
	\begin{align}
	V(s,x)=\sup_{(c,\pi, \Phi)\in \mathcal{A}(s,x)}\mathbb{E}\left[\int_0^{\tau}e^{-\beta t}\ln c_t dt+e^{-\beta \tau} F(X_{\tau})\Biggr | S_{0-} = s, X_{0-} = x\right]
	\label{eq:OptiProb}
	\end{align}
	due to dynamic programming principle [see \cite{jean-lakner-kadam04}].
		
	\begin{remark}
	\label{remark:tax}
	It is straightforward to introduce a dividend tax rate of $\kappa\in[0,1)$ in the model such that investors' wealth process, prior to firm's default, satisfies 
	\begin{align}
	dX_t=(rX_t-c_t)dt+(1-\kappa)d\Phi_t
	\label{eq:X_with_tax}
	\end{align}
	instead. Then under the transformation $\tilde{X}_t:=\frac{X_t}{1-\kappa}$ we can recover a version of the problem without tax.
	\end{remark}

	\begin{remark}
	\label{remark:agency}
	The optimization problem \eqref{eq:OptiProb} is a ``first best'' criteria where firm managers and investors can perfectly coordinate to jointly deduce the optimal corporate policies and consumption strategy to create maximum value for investors. It is indeed also possible to adopt our mathematical framework within an agency-based model featuring self-interested, risk averse managers as in Lambrecht and Myers (\citeyear{lambrecht-myers12}, \citeyear{lambrecht-myers17}). In this alternative setup, managers capture $1-\kappa$ fraction of the firm's total payout as a form of rent extraction and pass the remaining $\kappa$ fraction to investors as dividends. We could interpret $X=(X_t)_{t\geq 0}$ in \eqref{eq:X_with_tax} and $c=(c_t)_{t\geq 0}$ as the private wealth level and consumption rate of the managers respectively. The parameter $1-\kappa$ now reflects the bargaining power of the managers.\footnote{The sharing rule of the firm's payout can be justified by solving a repeated bargaining game between managers and investors. See Lambrecht and Myers (\citeyear{lambrecht-myers12}, \citeyear{lambrecht-myers17}).} They simultaneously set the corporate policies and their consumption strategy to maximize their own lifetime utility as in \eqref{eq:OptiProb}.
	\end{remark}

	It is constructive to compare our modeling setup to that of \cite{hobson-zhu16} who consider a Merton consumption and investment problem in which the risky asset can only be sold but not bought. Our problem is similar to theirs in the sense that the transfer of value also occurs in one direction only from the firm to the investors as dividends but capital cannot be injected into the firm from the investors. In other words, we rule out the possibility of equity issuance within our setup. It is not an unreasonable assumption as equity financing often involves expensive and time consuming procedures especially for smaller firms. 

	In \cite{hobson-zhu16} wealth is allocated between a risky asset and a riskfree cash account as in the standard Merton problem. Meanwhile, our model concerns wealth allocation between a risky firm and a riskfree cash account where the value of the former is not an exogenously given process but rather a controlled process based on the capital structure decision $\pi$. It is indeed possible to view our setup as a variant of the Merton problem with transaction costs [as considered by \cite{magill-constantinides76}, \cite{davis-norman90} and \cite{shreve-soner94}], albeit with a very special transaction costs structure. The economy features three distinct assets: a risky defaultable asset with drift $\mu$ and volatility $\sigma$, a defaultable debt instrument with constant yield $\rho$ and a non-defaultable cash account with interest rate $r$. The two defaultable securities can only be sold for cash but not bought with cash, i.e. transaction cost is infinite on purchase and zero on sale. However, these two securities are fully fungible which can be freely converted from one into another at their prevailing value without any friction. This unique transaction cost structure makes our problem non-trivial where the solution construction and the related economic properties do not follow from the existing literature of portfolio optimization under transaction costs.
	
	\section{Main results}
	\label{sect:mainresult}
	
	We state the key results of this paper where the proof is deferred to Section \ref{sect:veri}.
	
	\begin{thm}
		For the optimization problem \eqref{eq:OptiProb}:
		\begin{enumerate}
		\item If $\mu=\rho\leq \lambda +r$, the optimal strategy is to liquidate the firm immediately by distributing its entire equity to investors in form of dividends and then investors consume their wealth at a rate of $\beta$. The corresponding value function is given by
		\begin{align}
		V(s,x)=F(s+x)=\frac{1}{\beta}\ln (s+x)+\frac{1}{\beta}\left[\frac{r}{\beta}+\ln\beta-1\right].
		\label{eq:ValFunAlwaysSale}
		\end{align}
		
		\item If $\mu\neq \rho$ or $\mu=\rho>\lambda+r$, there exists a constant $z^*\in(0,\infty)$ such that the optimal strategy (prior to default of the firm) is not to pay any dividend when $\frac{S_t}{X_t}\leq z^*$. On this region, the optimal consumption strategy and investment policy are given by the feedback controls $c^{*}_t=c^*(S_t,X_t)$ and $\pi^{*}_t=\pi^*(S_t,X_t)$ where
		\begin{align}
		c^*(s,x):=\frac{1}{V_x(s,x)},\qquad \pi^*(s,x):=-\frac{(\mu-\rho)V_s(s,x)}{\sigma^2 sV_{ss}(s,x)}.
		\label{eq:FeedbackControls}
		\end{align}
		$V$ is the value function of the problem to be defined in Proposition \ref{prop:CandidateVProp} in Section \ref{sect:veri}.
		When $\frac{S_t}{X_t}> z^*$, the optimal strategy is to pay a discrete dividend of size $\frac{S_t-z^* X_t}{1+z^*}$ to the investors and then the strategies associated with the region of $\frac{S_t}{X_t}\leq z^*$ are followed thereafter. 
		
		\end{enumerate}
		
		\label{thm:main}
	\end{thm}
	
	
	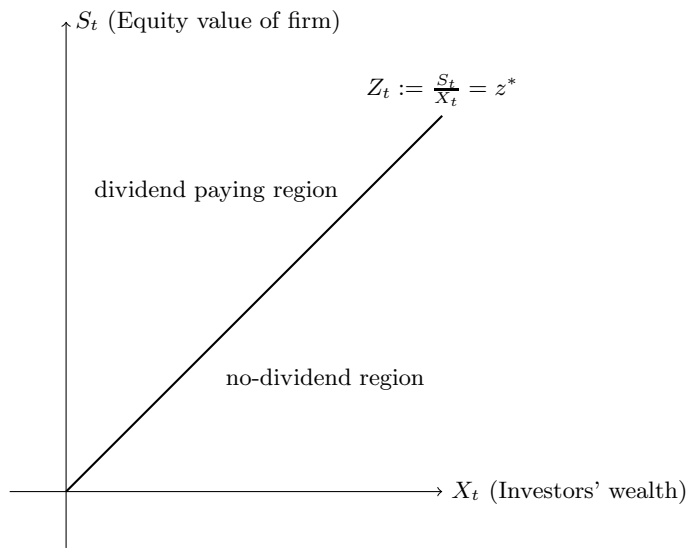
\begin{figure}[htbp]
		\center
		\begin{tikzpicture}[scale=5]
		\draw[->] (-0.15,0) -- (1,0) node [right] {\small $X_{t}$ (Investors' wealth)};
		\draw[->] (0,-0.15) -- (0,1.25) node[right] {\small $S_{t}$ (Equity value of firm)};
		\draw[thick][-] (0,0) -- (1,1) node[above] {\small $Z_{t}:=\frac{S_t}{X_t}=z^{*}$};
		\draw (0.4,0.3) node [right] {\small no-dividend region};
		\node at (0.4,0.8) {\small dividend paying region};
		\end{tikzpicture}
		\caption{A graphical illustration of the optimal dividend strategy. When the equity value of the firm is too high or investors' private wealth is too low such that $S_t/X_t> z^*$, a discrete amount of dividend is paid out by the firm to bring the ratio $S_t/X_t$ back to $z^*$. No dividend is paid when the state variables lie within the no-dividend region.}
	\label{fig:wedge}
	\end{figure}
	
	In our model, there are two economic motives for the equity investors to invest in the firm. The first motive is brought by the investment prospect of the firm which exists for as long as the excess return of the risky asset $\mu- \rho$ is non-zero. Note that investors are still willing to invest in the firm even if $\mu<\rho$ because it is possible for the firm to short sell the poor performing risky asset for value creation. The second motive lies within the funding advantage of the firm due to its access to corporate debt financing. At minimum, the firm can serve as a ``risky bank account'' with yield $\rho$ and default rate $\lambda$. This funding vehicle is superior to the investors' private saving account provided that the default-risk-adjusted yield $\rho-\lambda$ is larger than the retail saving rate $r$. For the parameter combination in part (1) of Theorem \ref{thm:main}, there is neither investment nor funding motive to invest in the firm and hence the optimal strategy is to liquidate the firm immediately.
	
	This very plausible prediction that a bad firm will be immediately liquidated is a unique feature relative to other models based on the standard Merton investment/consumption framework (such as \cite{lambrecht-myers17}). Under risk aversion, individuals demand a smooth consumption schedule. If dividends are tied with consumption, then risk aversion will force dividends to be smooth as well and hence a bad firm must be run continuously to generate cash flows over time. By distinguishing consumption and dividend via the possibility of depositing investors' private wealth in a retail account, individuals can opt to liquidate a bad firm and consume the proceeds over time optimally according to their own preference.
	
	For the more general parameter combination in part (2) of Theorem \ref{thm:main}, the optimal dividend strategy resembles the optimal investment strategy of a Merton problem with infinite transaction costs as in \cite{hobson-zhu16}. Rigorously speaking, the optimal dividend strategy $\Phi^*$ is a local time policy which keeps $Z_t:=\frac{S_t}{X_t} \leq z^*$, and it can be characterized by the solution to a Skorohod equation with reflecting boundary along $Z_t=z^*$. Dividends are paid when the firm value is too high or investors' private wealth level is too low in order to keep ratio of firm value to private wealth below a threshold. See Figure \ref{fig:wedge}. In particular, the payout trigger target is given by $S_t\leq z^* X_t$ where the right hand side is not a constant [as commonly seen in standard dividend distribution models such as \cite{radner-shepp96}] but instead it increases with the investors' private wealth. As investors become more wealthy, consumption can be adequately supported from their private account and hence a larger fraction of equity can be retained within the firm to further finance its investment activities. It is another unique feature of our model due to the disentanglement of dividend payout and consumption.
	
	\begin{remark}
	In the context of corporate finance, it might be more sensible to impose the constraint $\pi>0$ since there may not exist a realistic way for a firm to dis-invest in a project. This can be incorporated within our model, and Theorem \ref{thm:main} will then remain the same except that the conditions for case (1) and (2) shall be replaced by ``$\mu\leq\rho\leq \lambda+r$'' and ``$\mu>\rho$, or $\mu\leq \rho$ and $\rho>\lambda+r$'' respectively.
	\end{remark}

	Although the characterization of the optimal controls in Theorem \ref{thm:main} is somewhat abstract under the non-degenerate case (2), a lower bound of $z^*$ is available and the monotonicity of $\pi^*$ and $c^*$ with respect to the state variables can be established. Moreover, in the corner case of $\mu=\rho>\lambda+r$ the closed-form expressions of the optimal controls can indeed be derived. The results are summarized in the following two propositions which proofs are given in Appendix \ref{appsect:statevar}. 
	
	\begin{prop}
		
		Suppose $\mu\neq \rho$ and consider the optimal controls defined in case (2) of Theorem \ref{thm:main}:
		
		\begin{enumerate}
			
			\item The critical threshold of dividend payment $z^*$ satisfies  $\frac{(\rho-r-\lambda)^{+}}{\lambda}\leq z^*<\infty$.
			
			\item The optimal investment level $\pi^*(s,x)$ admits an expression $\pi^*(s,x)=\frac{\mu-\rho}{\sigma^2}\theta(s/x)$ where $\theta(z)$ is a bounded, positive and decreasing function on $0<z\leq z^*$ with $\theta(z^*)>1$.
			
			\item The optimal consumption rate per unit private wealth $\bar{c}(s,x):=\frac{c^*(s,x)}{x}$ is a function of $z=\frac{s}{x}$ only which is increasing on $0<z\leq z^*$ and $\bar{c}(z)\downarrow \beta$ as $z\downarrow 0$.
				
		\end{enumerate}
		
		\label{prop:policyrange}
	\end{prop}

	\begin{prop}
	
	Suppose $\mu= \rho>\lambda+r$ and consider the optimal controls defined in case (2) of Theorem \ref{thm:main}. Then
	\begin{align*}
	z^*=\frac{\rho-r-\lambda}{\lambda},\qquad \pi^*(s,x)=0,\qquad c^*(s,x)=\frac{x}{1/\beta-q(s/x)}
	\end{align*}
	where $q=q(z)$ is an implicit function defined as the solution to
	\begin{align}
	z=\left(\frac{\rho-r-\lambda}{\lambda}\right)^{\frac{\lambda}{\beta+\lambda}}\left(\frac{\rho-r-\lambda}{\beta(\rho-r)}\right)^{-\frac{\rho-r}{\beta+\lambda}}\left(\frac{\beta q}{1-\beta q}\right)^{\frac{\beta}{\beta+\lambda}}q^{\frac{\rho-r}{\beta+\lambda}}.
	\label{eq:closeform_zq}
	\end{align}
	\label{prop:PolicySpecialCase}
\end{prop}

Analytical expressions of all the important control variables are available in the case of $\mu=\rho>\lambda+r$ where the investment motive vanishes.\footnote{Note that the case of $\mu=\rho\leq \lambda + r$ has been covered by case (1) of Theorem \ref{thm:main} where we can take $z^*=0$ corresponding to immediate liquidation.} The firm essentially becomes a pure funding vehicle (subject to default risk) without any investment in the risky asset as $\pi^*=0$. The optimal dividend policy is then entirely driven by the funding quality of the firm measured by $\frac{\rho-r}{\lambda}$, which can be interpreted as the corporate yield spread per unit default risk. Higher this ratio, less often the firm pays dividends since it is more efficient to retain the capital within the firm for value creation.

In the more general case of $\mu\neq \rho$, part (1) of Proposition \ref{prop:policyrange} implies that $z^*\uparrow\infty$ as $\lambda\downarrow 0$ provided that $\rho>r$. The economic interpretation of this limiting result is the following: as long as the funding advantage $\rho>r$ exists, it is in general suboptimal to pay out any dividend because a unit of wealth in the investors' private account can only earn the retail rate $r$ whereas a unit of equity within the firm can at least earn a better rate $\rho$ when default of firm is not a concern. Dividends should thus only be paid out after investors' private wealth has been entirely depleted. Once $X$ hits zero, the firm pays out dividend continuously to match the optimal consumption required by the investors. Although we do not explicitly consider $\lambda=0$ in this paper, the optimal strategy of this special case can be characterized rigorously similar to Theorem 10 of \cite{hobson-zhu16}.

The magnitude of the firm's investment level is decreasing in equity value $S_t$ and increasing in investors' wealth $X_t$. To understand this behavior, it is important to note that investors' exposure to the risky asset is determined by two factors: the fraction of wealth invested in the firm $Z_t=S_t/X_t$ and the investment level of the firm $\pi_t$. In general, risk averse investors desire to maintain a target exposure to the risky asset. A benchmark example is that if there is no friction in the economy such that investors do not need to rely on the firm as an intermediary to invest in the risky asset, then they will invest a constant fraction of their wealth in the risky asset given by the Merton ratio. In our model where equity financing cannot be performed freely, the firm thus has to be leveraged more aggressively when equity value declines or investors' wealth increases to provide an efficient risk-return exposure for the investors. Notice that $\theta(z^*)>1$ implies the investment level is always larger than the Merton ratio (in magnitude). 
 
Consumption rate rises as the state variable $Z_t=S_t/X_t$ approaches the critical threshold $z^*$ because investors anticipate a dividend payment is due soon which will boost their wealth level and hence a larger consumption today becomes sustainable. 

The next proposition highlights the key comparative statics of our model. The proof can be found in Appendix \ref{appsect:compstat}.

\begin{prop}
		Suppose we are in the non-degenerate case of either $\mu\neq \rho$ or $\mu=\rho>\lambda+r$. Then the following comparative statics hold:
		\begin{enumerate}
			\item The critical threshold of dividend payment $z^*$ is decreasing in $\lambda$, $r$ and $\sigma$, and is increasing (resp. decreasing) in $\mu$ when $\mu\geq \rho$ (resp. $\mu\leq \rho$).
			\item The optimal investment level $\pi^*(s,x)$ is increasing (resp. decreasing) in $\lambda$ and $r$ when $\mu> \rho$ (resp. $\mu< \rho$).
			\item The optimal consumption level $c^*(s,x)$ is decreasing in $\lambda$ and $r$.

		\end{enumerate}
		\label{prop:compstat}
\end{prop}

\begin{figure}[!htbp]
	\captionsetup[subfigure]{width=0.5\textwidth}
	\centering
	\subcaptionbox{Dividend payment boundary.}{\includegraphics[scale =0.515] {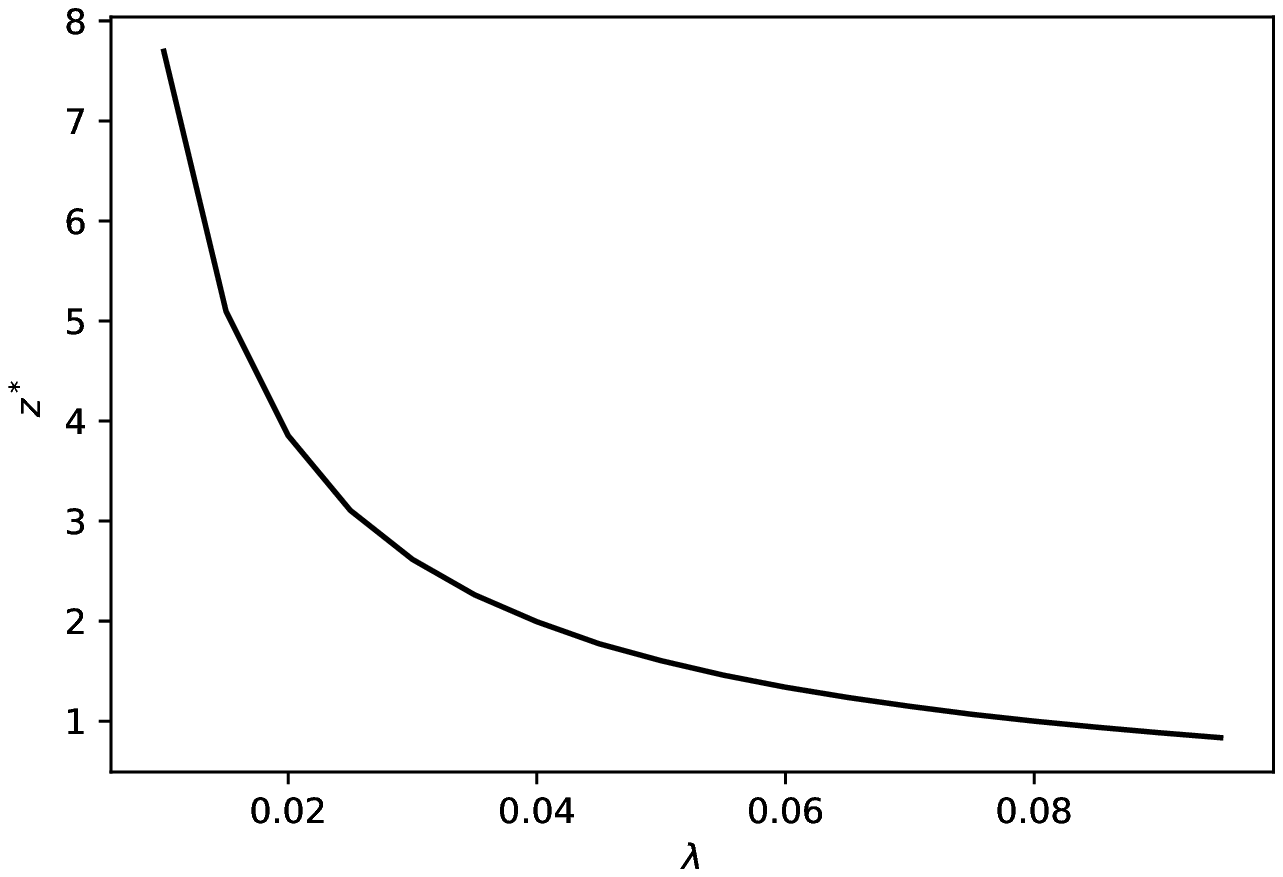}}
	\subcaptionbox{Optimal investment level.}{\includegraphics[scale =0.515]{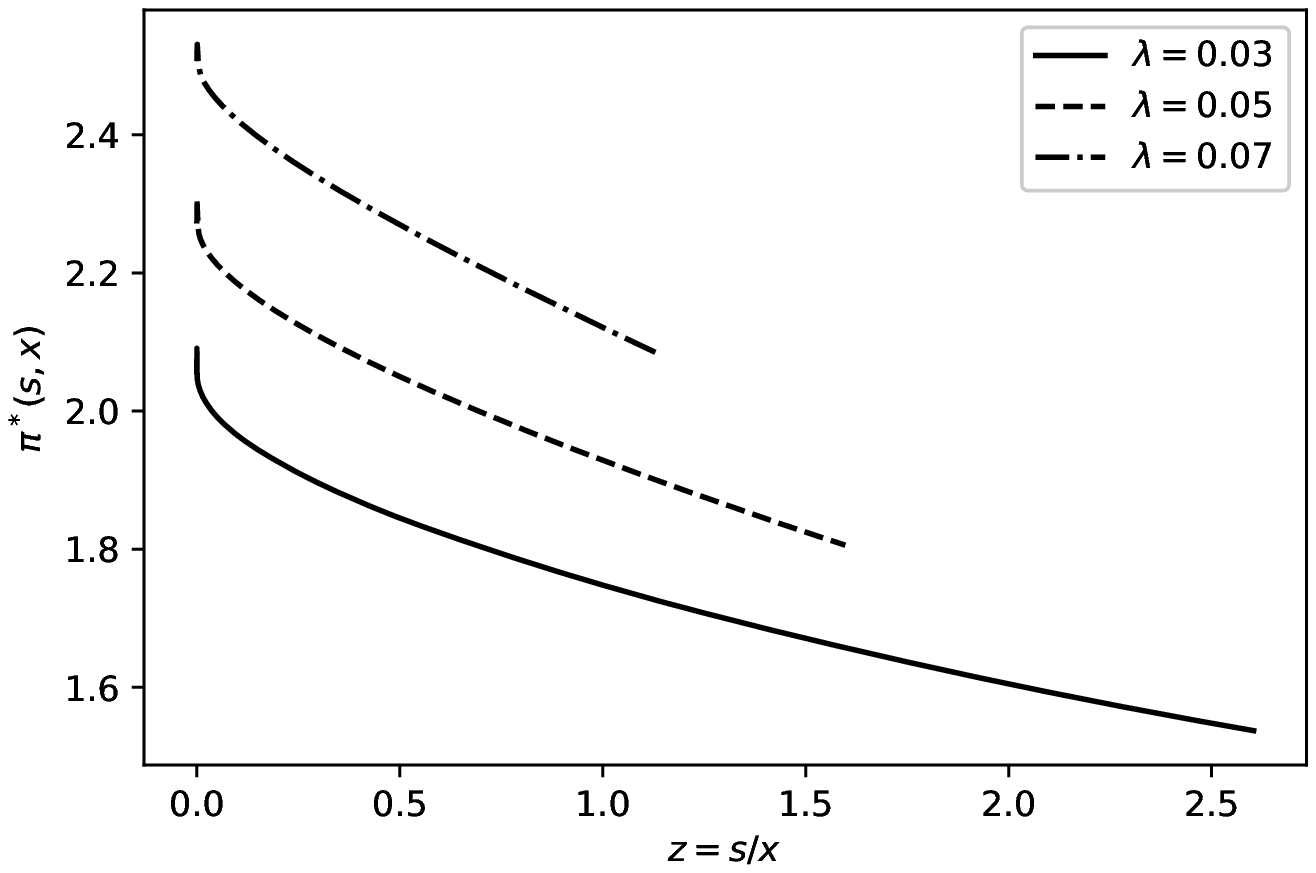}}
	\subcaptionbox{Optimal consumption.}{\includegraphics[scale =0.515]{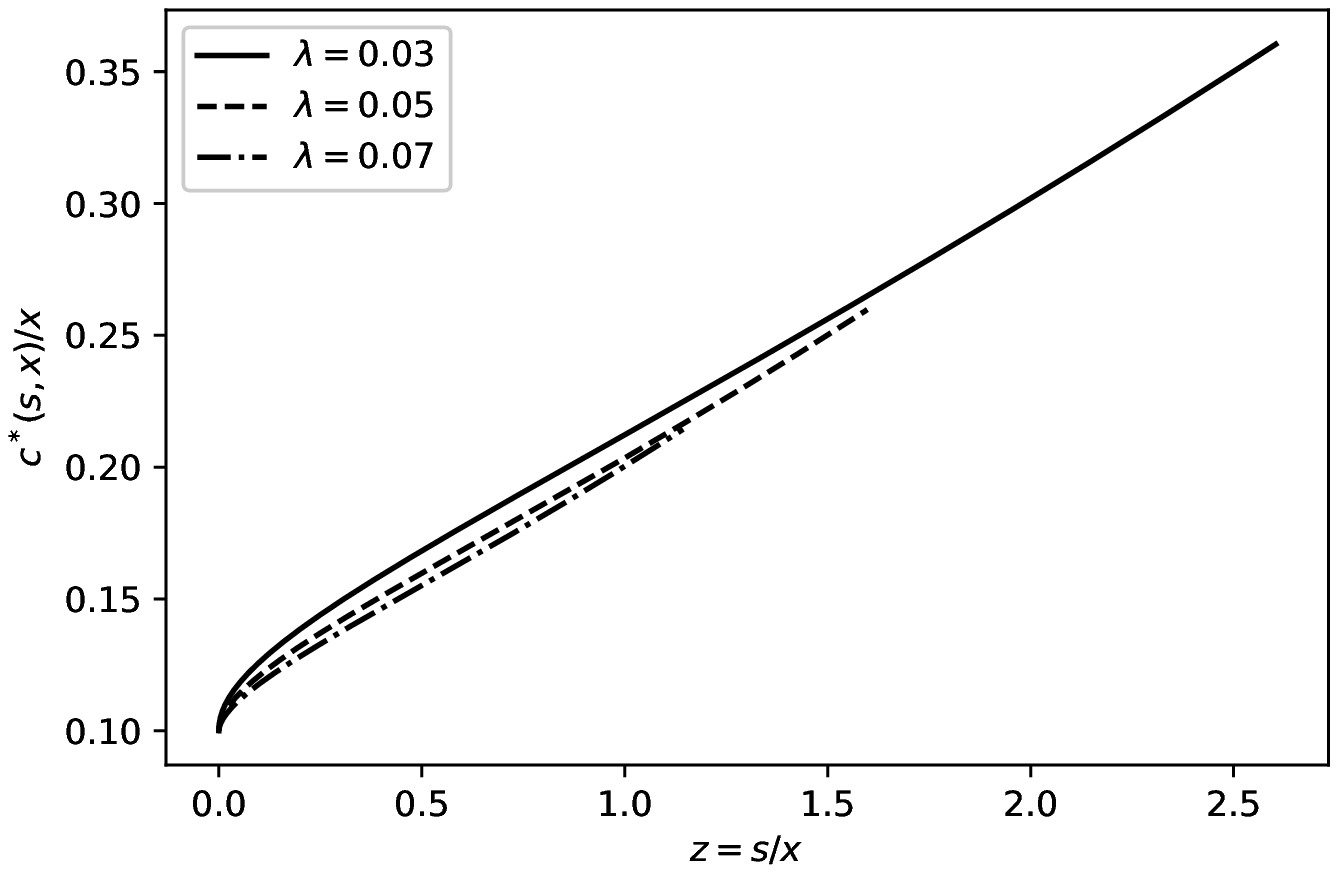}}
	
	\caption{Comparative statics of the optimal controls with respect to $\lambda$. Base parameters used are: $\mu=0.2$, $\sigma=0.3$, $\rho=0.1$, $r=0.02$ and $\beta=0.1$.}
	\label{fig:compstat}
\end{figure}

We focus on the comparative statics results with respect to $\lambda$ (see Figure \ref{fig:compstat} for some numerical illustrations). Firstly, with a higher default risk $\lambda$ the firm pays out dividends more aggressively as reflected by a lower critical threshold $z^*$. This can be understood as a phenomenon of moral hazard. When managers foresee that their firm is more likely to fail, there is a stronger precautionary incentive to transfer value within the firm to the investors' in form of dividends because the wealth in the private pocket of investors is left untouched due to limited liability while equity within the firm is seized or wiped out in the event of default.  

Among all the results in Proposition \ref{prop:compstat}, perhaps the most surprising one is that the magnitude of $\pi^*$ is increasing in $\lambda$. This may seem counter-intuitive in view of the risk averse nature of equity investors as the result here suggests the (absolute) investment level is higher for a riskier firm (in terms of default risk). Nonetheless, this result should be understood in conjunction with the impact of $\lambda$ on $z^*$. A higher $\lambda$ leads to a lower dividend threshold $z^*$ and as such a smaller pool of equity (relative to investors' wealth) is retained in the firm on average. To offset this effect of under-investment, the firm has be leveraged more aggressively to restore the overall exposure to the risky asset.

Finally, the effect of $\lambda$ on the optimal consumption is intuitive where a higher default risk of the firm encourages precautionary saving of the investors by reducing consumption.

The above results could be of interests to the area of corporate finance because they suggest aggressive payout policy and capital structure could potentially be driven by a high level of default risk where equity holders exploit the limited liability structure of the firm to retain value for themselves, thus a symptom of moral hazard. An important caveat behind the above interpretations is that we assume both $\mu$ and $\rho$ are fixed constants independent of $\lambda$. Implicitly, we assume asymmetric information between equity holders and bondholders. Only the former know the firm is subject to a Poisson default shock of intensity $\lambda$ while the latter do not take this default possibility into account when setting the debt yield. On the other hand, the default risk is assumed to be an idiosyncratic one which does not improve the risky asset return $\mu$. In a bank lending model of \cite{lambrecht-tse18}, both $\rho$ and $\mu$ are linked to $\lambda$ depending on how an insolvent bank is resolved during an economic downturn. In general, we could incorporate extensions of this kind by replacing $\rho$ and $\mu$ by some functions of $\lambda$. This will not significantly change the mathematical analysis of the stochastic control problem. However, the comparative statics with respect to $\lambda$ will now depend on the precise constructions of $\rho(\lambda)$ and $\mu(\lambda)$.
	
	\section{A heuristic derivation of the solution}
	\label{sect:heuristics}
	
	In this section we use heuristics to derive an equation that the value function should satisfy. Its optimality is then verified rigorously in Section \ref{sect:veri}.
	
	It is convenient to reformulate \eqref{eq:OptiProb} as an infinite horizon control problem. Note that
	\begin{align*}
	\mathbb{E}\left[\int_0^{\tau}e^{-\beta t}\ln c_t dt+e^{-\beta \tau} F(X_{\tau})\right] 
	&=\mathbb{E}\left[\int_0^{\infty}1_{(\tau>t)}e^{-\beta t}\ln c_t dt\right]+\mathbb{E}\left[e^{-\beta \tau} F(X_{\tau})\right] \nonumber\\
	&=\mathbb{E}\left[\int_0^{\infty}e^{-(\beta+\lambda) t}\ln c_t dt\right]+\int_0^\infty \lambda e^{-\lambda t}\mathbb{E}\left[e^{-\beta t} F(X_{t})\right]dt \nonumber\\
	&=\mathbb{E}\left[\int_{0}^\infty e^{-(\beta+\lambda)t}\left(\ln c_t+\lambda F(X_t)\right)dt\right]
	\end{align*}
	where we have used the fact that $\tau$ is an independent exponential random variable. Hence problem \eqref{eq:OptiProb} is equivalent to
	\begin{align}
	V(s,x)=\sup_{(c,\pi,\Phi)\in\mathcal{A}(s,x)}\mathbb{E}\left[\int_{0}^\infty e^{-(\beta+\lambda)t}\left(\ln c_t+\lambda F(X_t)\right)dt\Biggr | S_{0-} = s, X_{0-} = x\right]
	\label{eq:OptProbInfinite}
	\end{align}
	which is an infinite horizon problem with discount rate $\beta+\lambda$ and running reward $\ln c_t+\lambda F(X_t)$.
	
	Write
	\begin{align*}
	M_t:=\int_{0}^t e^{-(\beta+\lambda)u}[\ln c_u+\lambda F(X_u)]du+e^{-(\beta+\lambda)t}V(S_t,X_t). 
	\end{align*}
	Then we expect $M$ is a supermartingale in general, and is a true martingale under the optimal strategy. Suppose $V$ is a $C^{2\times 1}$ function, applying Ito's lemma we find
	\begin{align*}
	e^{(\beta+\lambda)t}dM_t&=\left\{\ln c_t-V_x c_t+r V_x X_t+\rho V_s S_t+(\mu-\rho)V_s S_t\pi_t+\frac{\sigma^2}{2}V_{ss} S_t^2\pi_t^2-(\beta+\lambda)V_t+\lambda F(X_t)\right\}dt\\
	&\qquad+(V_x-V_s)d\Phi_t+\sigma\pi_t V_s S_t dB_t.
	\end{align*}
	Further assume $V_x>0$ and $V_{ss}<0$, the drift term can be maximized with respect to $c$ and $\pi$. Then we expect the value function to solve the HJB variational inequality
	\begin{align}
	\min\left(-\mathcal{L}V, -\mathcal{M}V\right)=0
	\label{eq:HJBIneq}
	\end{align}
	where the operators $\mathcal{L}$ and $\mathcal{M}$ are defined as
	\begin{align}
	\mathcal{L}f&:=\sup_{c>0,\pi}\left\{\ln c-f_x c+r f_x x+\rho f_s s+(\mu-\rho)f_s s\pi+\frac{\sigma^2}{2}f_{ss} s^2\pi^2-(\beta+\lambda)f+\lambda F(x)\right\}\nonumber \\
	&=-\ln f_x-1+r f_x x+\rho f_s s-\frac{(\mu-\rho)^2 [f_s]^2}{2\sigma^2  f_{ss}}-(\beta+\lambda)f+\lambda F(x), \label{eq:LOperator} \\
	\mathcal{M}f&:= f_x-f_s.\label{eq:MOperator}
	\end{align}
	
	Inspired by \cite{magill-constantinides76}, we conjecture the optimal strategy is to pay dividends only when the ratio $Z_t:=\frac{S_t}{X_t}$ is above a certain threshold $z^*$ (refer to Figure \ref{fig:wedge} again). Following \cite{davis-norman90}, we postulate the value function has the form
	\begin{align}
	V(s,x)=\frac{1}{\beta}\ln x+g\left(\frac{s}{x}\right)
	\end{align}
	for some function $g$ to be determined. Write $z:=\frac{s}{x}$. Then over $z \leq z^*$ no dividend is paid and we expect $\mathcal{L}V=0$ which becomes
	\begin{align}
	-\ln \left(\frac{1}{\beta}-zg'(z)\right)+(\rho-r) zg'(z)-\frac{(\mu-\rho)^2[g'(z)]^2}{2\sigma^2 g''(z)} -(\beta+\lambda) g(z)+ \frac{\lambda}{\beta}\left(\frac{r}{\beta}+\ln\beta-1\right)+\frac{r}{\beta}-1=0.
	\label{eq:HJBg}
	\end{align}
	
	The system can be further simplified by a series of transformation used in \cite{hobson-zhu16} and \cite{hobson-tse-zhu18}. Write $u:=\ln z$ and let $h(u)=h(\ln z):= g(z)-\frac{1}{\beta+\lambda}\left[\frac{\lambda}{\beta}\left(\frac{r}{\beta}+\ln\beta-1\right)+\frac{r}{\beta}-1\right]$. Then
	\begin{align*}
	g'(z)&=\frac{d}{dz}g(z)=\frac{d}{dz}h(u)=\frac{d}{du}h(u)\frac{du}{dz}=\frac{h'(u)}{z}, \\
	g''(z)&=\frac{d}{dz}\frac{h'(u)}{z}=\frac{1}{z}\frac{d}{dz}h'(u)-\frac{h'(u)}{z^2}=\frac{h''(u)}{z^2}-\frac{h'(u)}{z^2}.
	\end{align*}
	\eqref{eq:HJBg} can then be reduced to
	\begin{align}
	-\ln \left(\frac{1}{\beta}-h'(u)\right)+(\rho-r) h'(u)-\frac{(\mu-\rho)^2[h'(u)]^2}{2\sigma^2 (h''(u)-h'(u))} -(\beta+\lambda) h(u)=0.
	\label{eq:HJBh}
	\end{align}
	Set $w(h):= \frac{dh}{du}$ such that $h''(u)=\frac{d}{du}w(h)=w'(h)h'(u)=w'(h)w(h)$. \eqref{eq:HJBh} then becomes
	\begin{align}
	-\ln \left(\frac{1}{\beta}-w(h)\right)+(\rho-r) w(h)-\frac{(\mu-\rho)^2 w(h)}{2\sigma^2 (w'(h)-1)} -(\beta+\lambda) h=0.
	\label{eq:HJBw}
	\end{align}
	Let $N$ be the inverse function of $w$, i.e. $N:= w^{-1}$, and write $q:=w(h)$. \eqref{eq:HJBw} can then be written as
	\begin{align}
	-\ln \left(\frac{1}{\beta}-q\right)+(\rho-r) q-\frac{(\mu-\rho)^2 q}{2\sigma^2 (1/N'(q)-1)} -(\beta+\lambda) N(q)=0.
	\label{eq:HJBN}
	\end{align}
	
	Suppose for now $\mu\neq\rho$. The special case of $\mu=\rho$ is discussed at the end of this section. Set $n(q):= \ln \left(\frac{1}{\beta}-q\right)+\beta N(q)$. After some algebra of substituting $N$ and $N'$ away by $n$ and $n'$ in \eqref{eq:HJBN}, we can obtain a first order ODE $n'(q)=O(q,n(q))$ where
	\begin{align}
	O(q,n):= \frac{\beta^2 q}{1-\beta q}\frac{m(q)-n}{n-\ell(q)}
	\label{eq:formO}
	\end{align}
	and
	\begin{align*}
	m(q)&:= \frac{\beta}{\beta+\lambda}\left\{(\rho-r) q+\frac{\lambda}{\beta}\ln\left(\frac{1}{\beta}-q\right)+\frac{(\mu-\rho)^2}{2\sigma^2\beta}\right\}, \\
	\ell(q)&:= \frac{\beta}{\beta+\lambda}\left\{\left[\rho-r+\frac{(\mu-\rho)^2}{2\sigma^2} \right]q+\frac{\lambda}{\beta}\ln\left(\frac{1}{\beta}-q\right)\right\}\\
	&=m(q)-\frac{\beta}{\beta+\lambda}\frac{(\mu-\rho)^2}{2\sigma^2}\left(\frac{1}{\beta}-q\right).
	\end{align*}
	
	We now derive the form of the value function on the dividend paying regime $\frac{s}{x}=z>z^*$. Under the conjectured strategy, a lump sum dividend $D$ is paid out by the firm to restore the equity value to investors' private wealth ratio back to $z^*$. $D$ should then solve $\frac{s-D}{x+D}=z^*$ which gives $D=\frac{s-z^*x}{1+z^*}$. The value function does not change on this corporate action and thus $V(s-D,x+D)=V(s,x)$ which is equivalent to
	\begin{align*}
	\frac{1}{\beta}\ln (x+D)+g(z^*)=\frac{1}{\beta}\ln x+g(z).
	\end{align*}
	From this we obtain
	\begin{align}
	g(z)&=\frac{1}{\beta}\ln\left(1+\frac{D}{x}\right)+g(z^*) 
	=\frac{1}{\beta}\ln\left(\frac{1+z}{1+z^*}\right)+g(z^*)  
	=\frac{1}{\beta}\ln(1+z)+A^*
	\label{eq:gfun_divregime}
	\end{align}
	on $z>z^*$ where $A^*$ is some constant.
	
	Now we apply the same set of transformations to \eqref{eq:gfun_divregime}. We have
	\begin{align*}
	h(u)&=g(e^u)-\frac{1}{\beta+\lambda}\left[\frac{\lambda}{\beta}\left(\frac{r}{\beta}+\ln\beta-1\right)+\frac{r}{\beta}-1\right] \\
	&=\frac{1}{\beta}\ln(1+e^u)+A^*-\frac{1}{\beta+\lambda}\left[\frac{\lambda}{\beta}\left(\frac{r}{\beta}+\ln\beta-1\right)+\frac{r}{\beta}-1\right]=:\frac{1}{\beta}\ln(1+e^u)+\bar{A}.
	\end{align*}
	Then
	\begin{align*}
	w(h)=h'(u)=\frac{1}{\beta}\frac{e^u}{1+e^u}=\frac{1}{\beta}\frac{e^{\beta (h-\bar{A})}-1}{e^{\beta( h-\bar{A})}}=\frac{1}{\beta}(1-e^{\beta(\bar{A}- h)})
	\end{align*}
	and the inverse function of $w$ is found to be
	\begin{align*}
	N(q)=w^{-1}(q)=\bar{A}-\frac{1}{\beta}\left(\ln \beta+\ln\left(\frac{1}{\beta}-q\right)\right).
	\end{align*}
	Finally,
	\begin{align}
	n(q)= \ln \left(\frac{1}{\beta}-q\right)+\beta N(q)=\beta\bar{A}-\ln\beta 
	&=\beta A^* -\left(\frac{r}{\beta}+\ln\beta -1\right)+\frac{\lambda}{\lambda+\beta}\ln\frac{1}{\beta} \nonumber \\
	&=\beta A^* -\left(\frac{r}{\beta}+\ln\beta -1\right)+\ell(0)
	\label{eq:nfun_div}
	\end{align}
	which is a constant. The above relationships hold as long as $z> z^*$, on which 
	\begin{align*}
	q= w(h)=\frac{1}{\beta}\frac{e^{u}}{1+e^{u}}=\frac{1}{\beta}\frac{z}{1+z}.
	\end{align*}
	Hence the equivalent range in the $q$-coordinate is $q>q^*:= \frac{1}{\beta}\frac{z^*}{1+z^*}$. 
	
	We expect the transformed value function $n$ to solve $n'=O(q,n)$ on $q\leq q^*$ and to be a constant on $q>q^*$. To solve such a free boundary value problem, we further require an initial value associated with the system. Along the boundary $s=0$ the equity value of the firm is zero and hence cannot invest in the risky asset nor the bond (if a firm with zero net equity attempts to borrow to invest or short sell the asset to support purchase of the bond, the Brownian nature of the asset price will make it impossible for the firm to maintain non-negative equity value). Then essentially the firm ceases to exist and the only feasible strategy is for the investors to consume their existing private wealth optimally. Hence $V(0,x)=F(x)=\frac{1}{\beta}\ln x+\frac{1}{\beta}\left[\frac{r}{\beta}+\ln\beta-1\right]$. This boundary condition translates into $g(0)=\frac{1}{\beta}\left[\frac{r}{\beta}+\ln\beta-1\right]$, $h(-\infty)=g(0)-\frac{1}{\beta+\lambda}\left[\frac{\lambda}{\beta}\left(\frac{r}{\beta}+\ln\beta-1\right)+\frac{r}{\beta}-1\right]=\frac{\ln\beta}{\beta+\lambda}$ and we expect $h'(-\infty)=0$. Then $w\left(\frac{\ln\beta}{\beta+\lambda}\right)=w(h(-\infty))=h'(-\infty)=0$, $N(0)=\frac{\ln\beta}{\beta+\lambda}$. and finally $n(0)=\ln\frac{1}{\beta}+\frac{\beta}{\beta+\lambda}\ln \beta =\frac{\lambda}{\beta+\lambda} \ln\frac{1}{\beta}=\ell(0)$.
	
	In summary, we are looking for a solution $n$ with initial value $n(0)=\ell(0)$ which solves $n'=O(q,n)$ on $0\leq q \leq q^*$ and $n(q)$ being a constant on $q>q^*$, where $q^*$ is an unknown boundary to be identified. The conjectured second order smoothness of the original value function $V$ now translates into the first order smoothness of $n$ and as such we expect $n'(q^*)=0$. But the form of \eqref{eq:formO} suggests that, away from $q=0$, $n'=O(q,n)=0$ if and only if $n=m$. Hence the boundary point $q^*$ must be given by the $q$-coordinate where the solution $n$ intersects the function $m(q)$.
	
	The next proposition confirms that a solution $n$ and the corresponding free boundary point $q^*$ indeed exist.
	\begin{prop}
		Suppose $\mu\neq \rho$. Consider an initial value problem
		\begin{align*}
		n'(q)=O(q,n(q)):= \frac{\beta^2 q}{1-\beta q}\frac{m(q)-n(q)}{n(q)-\ell(q)},\qquad n(0)=\ell(0) \text{ and }n'(0)>\ell'(0)
		\end{align*}
		where
		\begin{align*}
		m(q)&:= \frac{\beta}{\beta+\lambda}\left\{(\rho-r) q+\frac{\lambda}{\beta}\ln\left(\frac{1}{\beta}-q\right)+\frac{(\mu-\rho)^2}{2\sigma^2\beta}\right\}, \\
		\ell(q)&:= \frac{\beta}{\beta+\lambda}\left\{\left[\rho-r+\frac{(\mu-\rho)^2}{2\sigma^2} \right]q+\frac{\lambda}{\beta}\ln\left(\frac{1}{\beta}-q\right)\right\}.
		\end{align*}
		
		A unique solution to the above problem exists at least up to $0\leq q \leq q^*$ with $n'(0)>0$ and $q^*:=\inf\{q> 0: n(q)\geq m(q)\}\in(\frac{1}{\beta}\frac{(\rho-r-\lambda)^{+}}{(\rho-r-\lambda)^{+}+\lambda},\frac{1}{\beta})$. The solution $n$ is strictly increasing and lies between $\ell(q)$ and $m(q)$ on $(0,q^*)$.
		
		\label{prop:theODE}
	\end{prop}
	
	\begin{proof}
		
		We first show that a unique solution exists in the neighborhood of $q=0$. Let $\chi(q):=n(q)-\ell(q)$. Then the initial value problem is equivalent to
		\begin{align}
		\chi'(q)&=n'(q)-\ell'(q)=O(q,\chi(q)+\ell(q))-\ell'(q) \nonumber \\
		&=\frac{\beta^2}{1-\beta q}\left[\frac{\beta}{\beta+\lambda}\frac{(\mu-\rho)^2}{2\sigma^2}\left(\frac{1}{\beta}-q\right)-\chi(q)\right]\frac{q}{\chi(q)}-\frac{\beta}{\beta+\lambda}\left(\rho-r+\frac{(\mu-\rho)^2}{2\sigma^2}-\frac{\lambda}{1-\beta q}\right) \nonumber \\
		&=:A(q,\chi(q))\frac{q}{\chi(q)}+B(q):=:\hat{O}(q,\chi(q))
		\label{eq:chiODE}
		\end{align}
		subject to $\chi(0)=0$ and $\chi'(0)>0$.
		
		Consider first a simpler problem in form of
		\begin{align*}
		\chi'(q)=A\frac{q}{\chi(q)}+B,\qquad \chi(0)=0 \text{ and }\chi'(0)>0
		\end{align*}
		where $A,B$ are constants with $A>0$. Making use of the substitution $y=\frac{\chi}{q}$, we can obtain an ODE in terms of $y=y(q)$ as $y+q\frac{dy}{dq}=\frac{A}{y}+B$ and in turn
		\begin{align}
		\frac{y}{A+By-y^2}dy=\frac{1}{q}dq.
		\label{eq:ODEFraction}
		\end{align}
		Since $A>0$, the denominator on the left hand side of \eqref{eq:ODEFraction} admits an expression of $(\alpha-y)(\gamma+y)$ for some $\alpha,\gamma > 0$. \eqref{eq:ODEFraction} can then be solved via partial fraction which leads to the solution of $|q|^{\alpha+\gamma}|\alpha-y|^{\alpha}|\gamma+y|^{\gamma}=C$ for some arbitrary constant $C$, or equivalently
		\begin{align*}
		|\alpha q-\chi|^{\alpha}|\gamma q+\chi|^{\gamma}=C.
		\end{align*}
		The initial condition $\chi(0)=0$ forces $C=0$. Hence the solution is either $\chi(q)=\alpha q$ or $\chi(q)=-\gamma q$ where the former satisfies the initial condition $\chi'(0)>0$.
		
		Return to the original problem \eqref{eq:chiODE}, as $A(0,0)=\frac{\beta^2}{\beta+\lambda}\frac{(\mu-\rho)^2}{2\sigma^2}>0$ a small extension to the above argument shows that a unique solution to problem \eqref{eq:chiODE} exists satisfying $\lim_{q\downarrow 0}\frac{\chi(q)}{q}=\alpha$ where $\alpha$ is the positive root to the quadratic equation
		\begin{align}
		A(0,0)+B(0)y-y^2=0.
		\label{eq:QuadEq}
		\end{align}
		The conclusion that $n'(0)>0$ is now clear since an application of L'Hopital's rule to $n'(q)=O(q,n(q))$ around $q=0$ gives
		\begin{align*}
		n'(0)=\frac{\beta^2[m(0)-\ell(0)]}{n'(0)-\ell'(0)}=\frac{\beta^2}{\beta+\lambda}\frac{(\mu-\rho)^2}{2\sigma^2}\frac{1}{n'(0)-\ell'(0)}
		\end{align*}
		and thus $n'(0)>\ell'(0)$ implies $n'(0)>0$.
		
		Away from the singular initial point $(0,0)$, the existence and uniqueness of the solution to the ODE are guaranteed by standard theories. Since $n'(0)>\ell'(0)$ and $n(0)=\ell(0)<m(0)$, $n$ is initially lying between $\ell(q)$ and $m(q)$. It is trivial from the form of $O(q,n)$ that the solution $n$ cannot cross $\ell$, and that the solution is increasing for as long as $q<\frac{1}{\beta}$ and $n(q)$ stays between $m(q)$ and $\ell(q)$. As $m(q)\to -\infty$ when $q\to \frac{1}{\beta}$, $n$ must cross $m$ somewhere on $q<\frac{1}{\beta}$ which guarantees the existence of $q^*:=\inf\{q> 0: n(q)\geq m(q)\}$. Moreover, as $n(q)<m(q)$ on $q<q^*$ and the derivative of $n$ has to be zero when $n$ crosses $m$, we must have $m'(q^*)\leq 0$. A simple calculus exercise shows that $m(q)$ has an inverted U-shape when $\rho-r-\lambda>0$ with its maximum attained at $q=\frac{\rho-r-\lambda}{\beta(\rho-r)}$ and thus $q^*\geq \frac{\rho-r-\lambda}{\beta(\rho-r)}$. Otherwise if $\rho-r-\lambda\leq 0$ then $m(q)$ is decreasing for all $0\leq q<1/\beta$ and in this case we can only conclude $q^*>0$. Combining these two cases leads to $q^*\in(\frac{1}{\beta}\frac{(\rho-r-\lambda)^{+}}{(\rho-r-\lambda)^{+}+\lambda},\frac{1}{\beta})$.
		
	\end{proof}
	
	\begin{remark}
	Note that we have imposed an additional constraint of $n'(0)>\ell'(0)$ in Proposition \ref{prop:theODE}. If we instead pick the solution with $n'(0)<\ell'(0)$, then $n(q)$ will be initially below $\ell(q)$. The form of $O$ suggests that $n$ is decreasing and does not cross $\ell$, and in turn $m$, for all $q$. Then there does not exist a boundary point $q^*$ at which smooth pasting holds. The resulting $n$ therefore is not a sensible candidate solution.
	\end{remark}
	
	Finally, we consider the special case of $\mu=\rho$. We can indeed obtain an explicit solution for $n$ where \eqref{eq:HJBN} gives a closed-form expression of $N$ as
	\begin{align*}
	N(q)=\frac{1}{\beta+\lambda}\left[(\rho-r)q-\ln\left(\frac{1}{\beta}-q\right)\right]
	\end{align*}
	and thus
	\begin{align*}
	n(q)=\ln\left(\frac{1}{\beta}-q\right)+\beta N(q)=\frac{\beta}{\beta+\lambda}\left[(\rho-r)q+\frac{\lambda}{\beta}\ln\left(\frac{1}{\beta}-q\right)\right]=m(q)=\ell(q).
	\end{align*}
	This result should not be surprising. From Proposition \ref{prop:theODE}, $\ell(q)<n(q)<m(q)$ on $0<q<q^*$ and $m(q)-\ell(q)\to 0$ as $\mu\to\rho$, we must have $n(q)\to m(q)=\ell(q)$ when $\mu$ approaches $\rho$.
	
	What should be the correct value of $q^*$ when $\mu=\rho$? Again, we expect $n$ is a constant on $q>q^*$ and first order smoothness suggests $q^*$ should satisfy $n'(q^*)=m'(q^*)=0$. Thus $q^*$ should be the $q$-coordinate of the turning point of $m(q)$ if it exists. The existence condition is given by $m'(0)> 0$ which is equivalent to $\rho-r-\lambda> 0$ and the corresponding value of $q^*$ is $\frac{\rho-r-\lambda}{\beta(\rho-r)}$. 
	
	If $\rho-r-\lambda\leq 0$, then one cannot locate any positive boundary point at which the first order smoothness holds. Our conjecture in this case is that the no-dividend region vanishes which is economically equivalent to $q^*=0$. The firm is liquidated immediately at time zero and investors receive the entire equity of the firm as dividends. The value function should thus satisfy
	\begin{align*}
	V(s,x)=F(s+x)&=\frac{1}{\beta}\ln (s+x)+\frac{1}{\beta}\left[\frac{r}{\beta}+\ln\beta-1\right]\\
	&=\frac{1}{\beta}\ln x+\frac{1}{\beta}\ln (1+z)+\frac{1}{\beta}\left[\frac{r}{\beta}+\ln\beta-1\right]
	\end{align*}
	and thus $g(z)=\frac{1}{\beta}\ln (1+z)+\frac{1}{\beta}\left[\frac{r}{\beta}+\ln\beta-1\right]$. If we apply the transformation which takes \eqref{eq:gfun_divregime} to \eqref{eq:nfun_div}, we can obtain $n(q)=\frac{\lambda}{\lambda+\beta}\ln\frac{1}{\beta}=m(0)=\ell(0)$ for all $q\geq 0$.
	
	To summarize this section, all possible shapes of $n$, $m$ and $\ell$ under different parameter combinations are shown in Figure 2.
	
	\begin{figure}[!htbp]
		\captionsetup[subfigure]{width=0.5\textwidth}
		\centering
		\subcaptionbox{$\mu\neq \rho$ and $\rho>r+\lambda$.}{\includegraphics[scale =0.375] {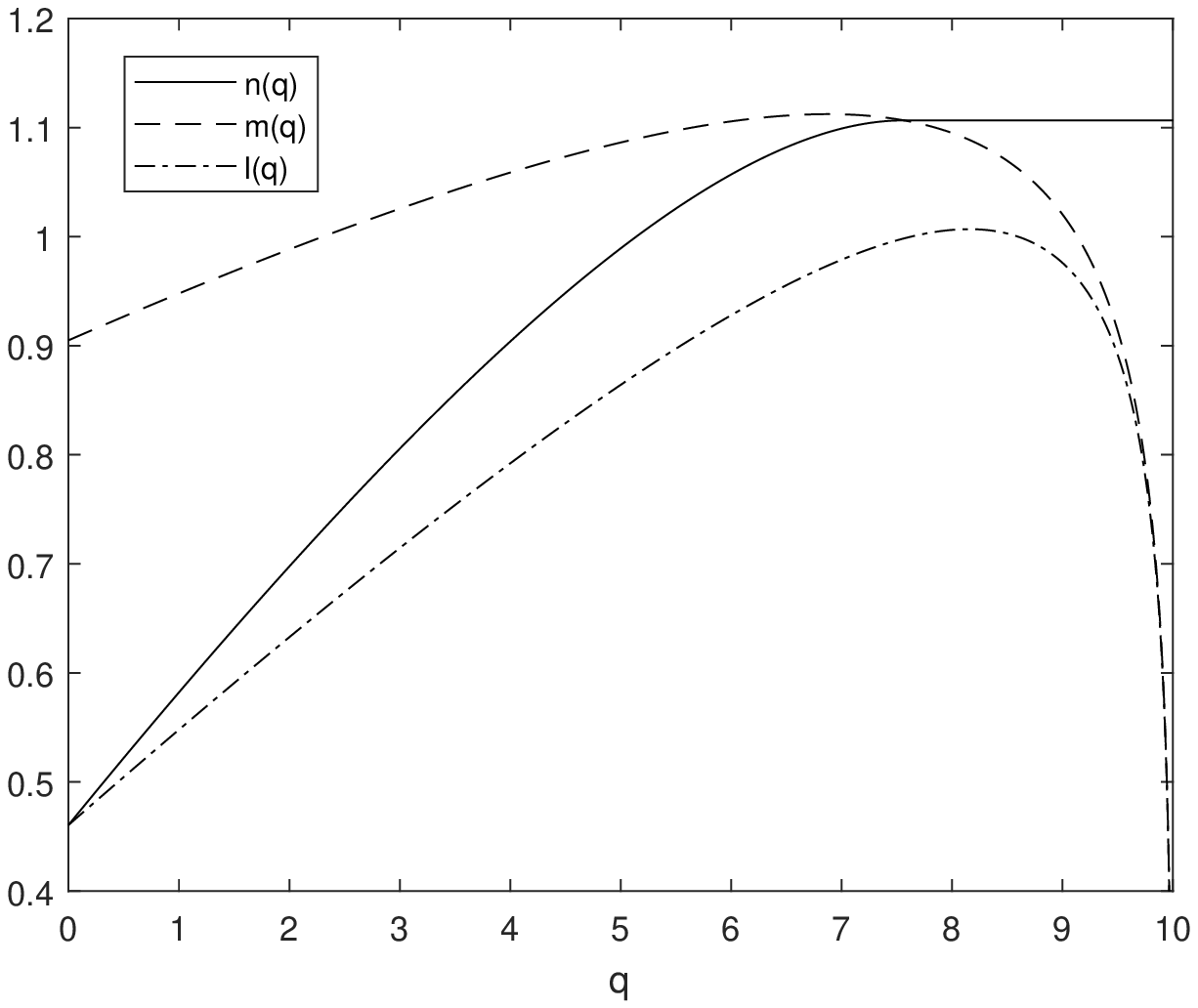}}
		\subcaptionbox{$\mu\neq \rho$ and $\rho\leq r+\lambda$.}{\includegraphics[scale =0.375]{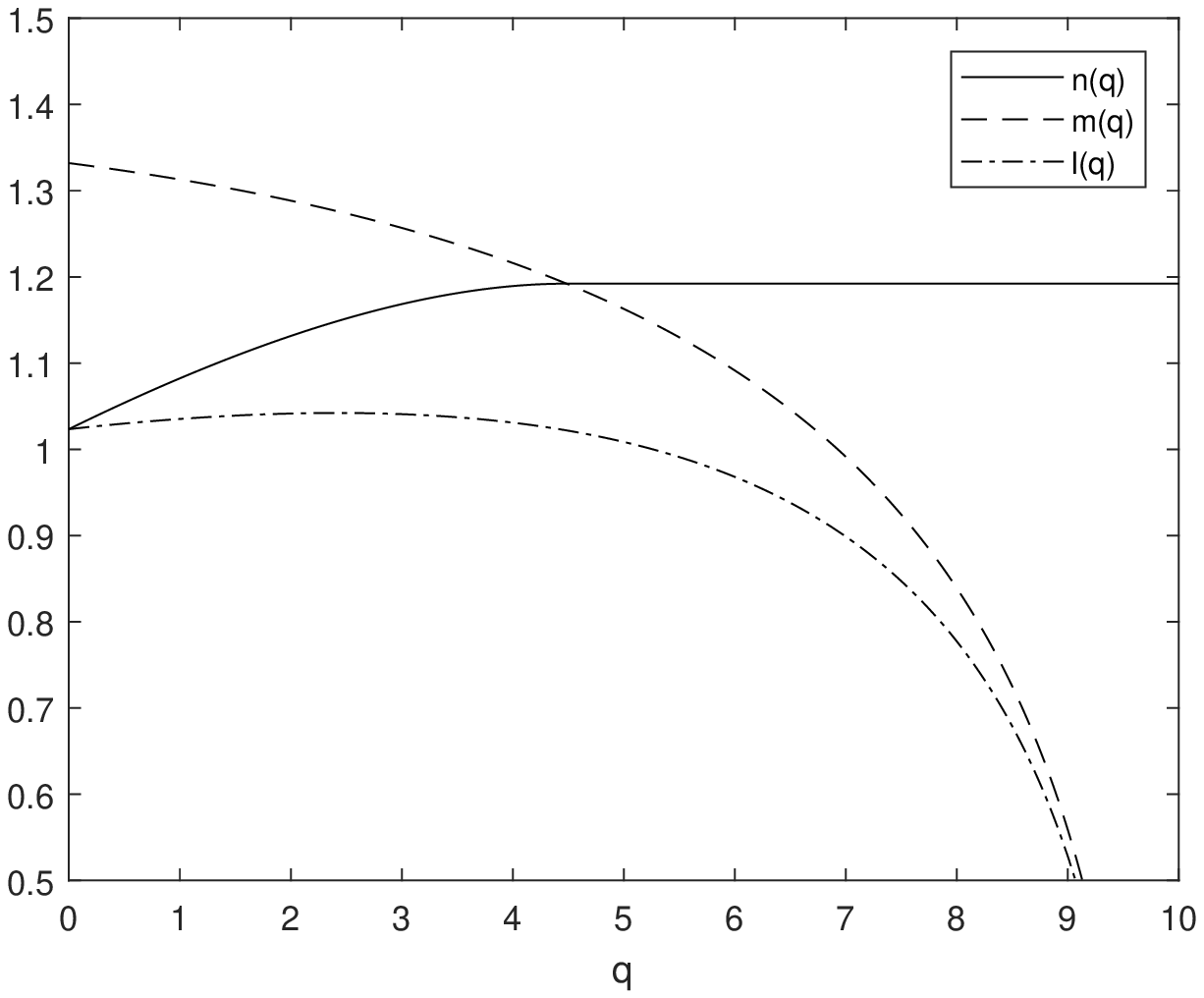}}
		\subcaptionbox{$\mu=\rho>r+\lambda$.}{\includegraphics[scale =0.375]{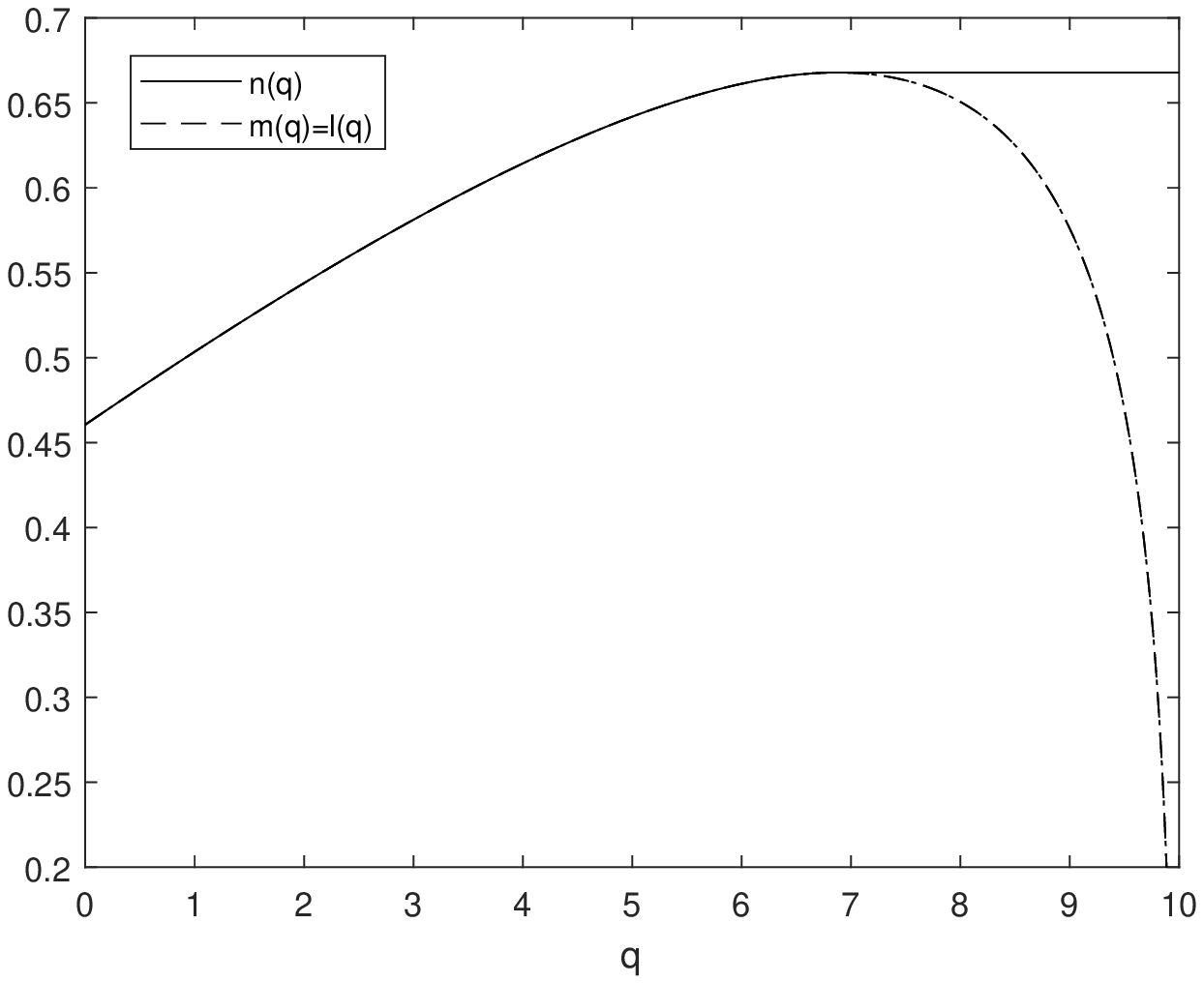}}
		\subcaptionbox{$\mu=\rho\leq r+\lambda$.}{\includegraphics[scale =0.375]{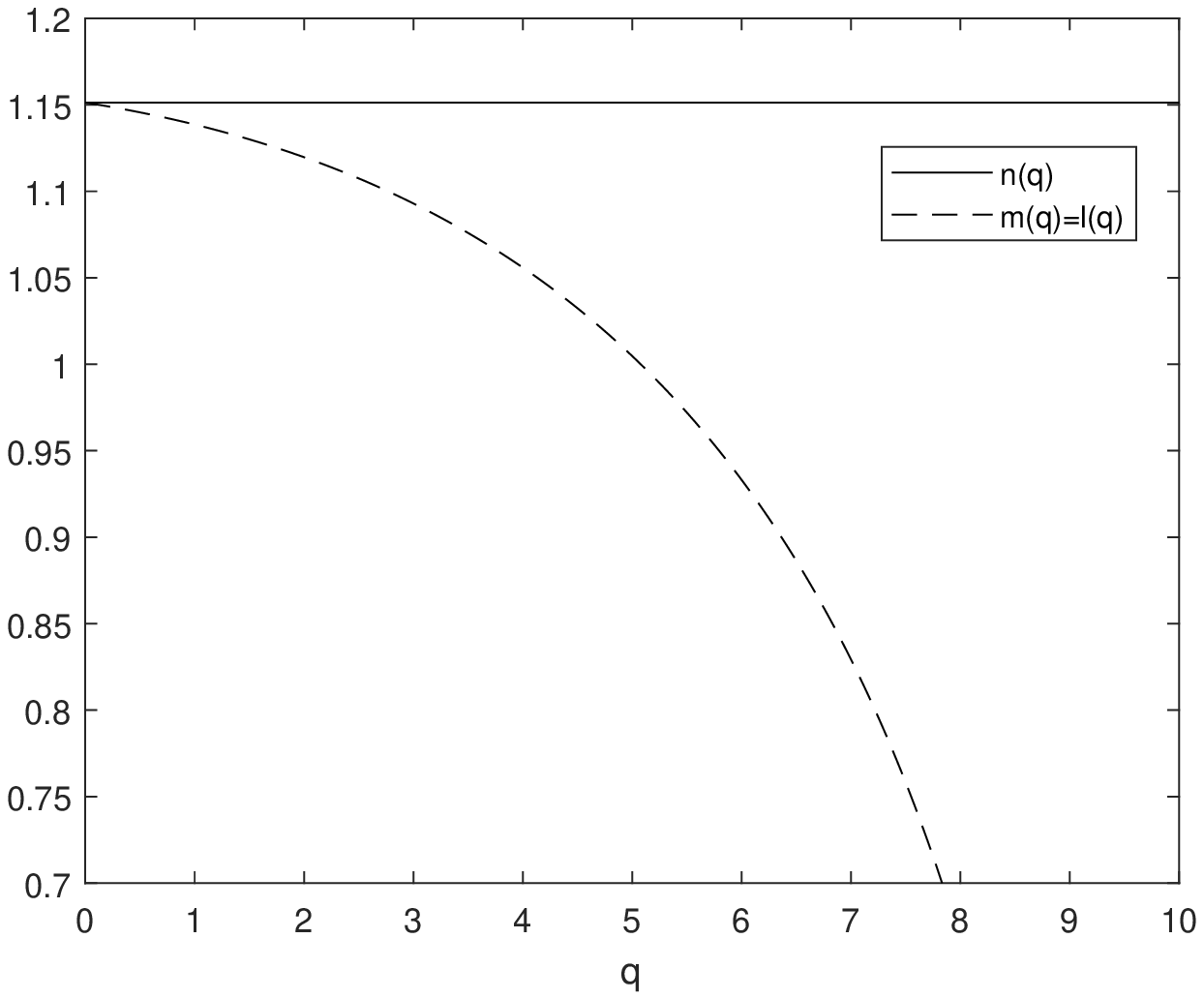}}
		
		\caption{The possible shapes of the transformed value function $n$ under different parameter combinations. When $\mu\neq \rho$, $n$ is first increasing and then becomes flat when it crosses $m$ at $q=q^*$. When $\mu=\rho$, $n$, $m$ and $\ell$ coincide and $q^*$ is the turning point of $m$ if exists. If $\mu=\rho\leq r+\lambda$ such that $m$ is decreasing for all $q\geq 0$, we take $q^*=0$ and the transformed value function is a flat horizontal line $n(q)=m(0)=\ell(0)$.} 
		\label{fig:shape_n}
	\end{figure}
	
	\section{Construction of the candidate value function and verification}
	\label{sect:veri}
	
	The heuristics in Section \ref{sect:heuristics} guide us to write down a first order system that the transformed value function should satisfy (with closed-form expressions available in some special cases). Conversely, given the solution to the first order system we can reverse the transformation to construct a second order smooth candidate value function and prove its optimality via a formal verification argument.
	
	We first construct the candidate value function in the special case of $\mu=\rho\leq r+\lambda$ in which case we expect the optimal strategy is to liquidate the firm immediately by transferring all equity to investors via dividends payment.
	\begin{prop}
	Suppose $\mu=\rho\leq r+\lambda$. On $(s,x)\in\mathbb{R}^2_{+}\setminus \{(0,0)\}$ define
	\begin{align}
	V^C(s,x)=\frac{1}{\beta}\ln(s+x)+\frac{1}{\beta}\left[\frac{r}{\beta}+\ln\beta-1\right].
	\label{eq:ValFunDegen}
	\end{align}
	Then $V^C(s,x)$ is a concave function increasing in both $s$ and $x$. Moreover, $\mathcal{L}V^C\leq 0$ and $\mathcal{M}V^C=0$, where $\mathcal{L}$ and $\mathcal{M}$ are the operators defined in \eqref{eq:LOperator} and \eqref{eq:MOperator}.
	\label{prop:ValFunPropDegen}
	\end{prop}

\begin{proof}
It is trivial that $V^C$ is concave and is increasing in both of its arguments. Direct computation gives $V^C_s=V^C_x=\frac{1}{\beta(s+x)}$ and in turn $\mathcal{M}V^C=0$. Finally,
\begin{align*}
\mathcal{L} V^C&=-\ln\frac{1}{\beta(s+x)}-1+\frac{r}{\beta}\frac{x}{s+x}+\frac{\rho}{\beta}\frac{s}{s+x}-(\beta+\lambda)\left[\frac{1}{\beta}\ln(s+x)+\frac{1}{\beta}\left(\frac{r}{\beta}+\ln\beta-1\right)\right] \\
&\qquad +\lambda\left[\frac{1}{\beta}\ln x+\frac{1}{\beta}\left(\frac{r}{\beta}+\ln\beta-1\right)\right] \\
&=\frac{\lambda}{\beta}\ln\frac{x}{s+x}+\frac{\rho-r}{\beta}\frac{s}{s+x}
\leq \frac{\lambda}{\beta}\ln\frac{x}{s+x}+\frac{\lambda}{\beta}\frac{s}{s+x}=\frac{\lambda}{\beta}\left(\ln\frac{x}{s+x}+1-\frac{x}{s+x}\right).
\end{align*}
Simple calculus exercise shows that $f(\alpha):=\ln\alpha +1-\alpha \leq 0$ for all $\alpha\in(0,1)$ and hence $\mathcal{L}V^C\leq 0$.
\end{proof}

Away from the special case of $\mu=\rho\leq r+\lambda$, we cannot write down the candidate value function explicitly. In the following two propositions, we describe how the transformation introduced in Section \ref{sect:heuristics} can be reversed and several important analytical properties of the constructed candidate value function are provided.
\begin{prop}
		
If $\mu\neq \rho$, let $n=(n(q))_{0\leq q\leq q^*}$ be the solution to the initial value problem in Proposition \ref{prop:theODE} where $q^*:=\inf\{q>0:n(q)\geq m(q)\}$. Otherwise if $\mu=\rho>r+\lambda$, define $n(q):=m(q)$ on $0\leq q\leq q^*:=\frac{\rho-r-\lambda}{\beta(\rho-r)}$. 

In both case, let $z^*:=\frac{\beta q^*}{1-\beta q^*}$ and $N(q):=\frac{1}{\beta}\left[n(q)-\ln\left(\frac{1}{\beta}-q\right)\right]$. Let $w$ be the inverse function of $N$. For $-\infty<u\leq \ln z^*$, define $h=h(u)$ as the solution to
		\begin{align}
		\int_{h}^{N(q^*)}\frac{dv}{w(v)}=\ln z^* - u
		\label{eq:defH2}
		\end{align}
		which is equivalent to
		\begin{align}
		\int_{w(h)}^{q^*}\frac{N'(v)}{v}dv=\ln z^* - u.
		\label{eq:defH}
		\end{align}
		$u\to h(u)$ is then a strictly increasing bijection from $(-\infty,\ln z^*]\to(N(0),N(q^*)]$. 
		
		Finally, set
		\begin{align}
		g^{C}(z):=
		\begin{cases}
		h(\ln z)+\frac{1}{\beta+\lambda}\left[\frac{\lambda}{\beta}\left(\frac{r}{\beta}+\ln\beta-1\right)+\frac{r}{\beta}-1\right],& 0< z\leq z^* ;\\
		\frac{1}{\beta}\ln(1+z)+\frac{1}{\beta}\left[n(q^*)+\left(\frac{r}{\beta}+\ln\beta -1\right)-\ell(0)\right] ,& z>z^*.
		\end{cases}
		\label{eq:CandidateG}
		\end{align}
		Then $g^C:(0,\infty)\to\mathbb{R}$ is a $C^2$ function.
		\label{prop:CandidateG}
	\end{prop}
	
	\begin{proof}
		
		For ease of notation we suppress the superscript $C$ in $g^C$ throughout the proof. We give the proof in the case of $\mu\neq\rho$. It is much easier to establish the results under $\mu=\rho$ since $n$ and $q^*$ are then available in closed-form.
		
		We first show that $h(u)$ is an increasing bijection. Recall from Proposition \ref{prop:theODE} that $n$ is increasing. Then $N'(v)=\frac{n'(v)}{\beta}+\frac{1}{1-\beta v}>0$ on $v<\frac{1}{\beta}$ and in turn the integrand on the left hand side of \eqref{eq:defH} is strictly positive such that $w(h(u))$ is strictly increasing in $u$. Moreover,
		\begin{align*}
		\int_{0+}^{\cdot}\frac{N'(v)}{v}dv>\int_{0+}^{\cdot}\frac{dv}{v(1-\beta v)}=+\infty.
		\end{align*}
		Hence $u\to w(h(u))$ is a bijection from $(-\infty,\ln z^*]$ to $(0,q^*]$. Since $N=w^{-1}$ is strictly increasing on $(0,q^*]$, $u\to h(u)$ is a bijection from $(-\infty,\ln z^*] \to (N(0),N(q^*)]$. 
		
		Now we proceed to show that $g$ is a $C^2$ function. On $z>z^*$, $g$ is trivially a $C^2$ function. On $0<z<z^*$, $n$ is a $C^1$ function and the continuity property is inherited by $(N,N')$ and then on integration by $(h,h',h'')$ and finally $(g,g',g'')$. It is thus sufficient to check the continuity of $g$, $g'$ and $g''$ at $z=z^*>0$. 
		
		From \eqref{eq:defH2}, $h(\ln z^*)=N(q^*)=\frac{1}{\beta}\left[n(q^*)-\ln\left(\frac{1}{\beta}-q^*\right)\right]=\frac{n(q^*)}{\beta}-\frac{1}{\beta}\ln\left(\frac{1}{\beta(1+z^*)}\right)$. Hence
		\begin{align*}
		g(z^*)&=h(\ln z^*)+\frac{1}{\beta+\lambda}\left[\frac{\lambda}{\beta}\left(\frac{r}{\beta}+\ln\beta-1\right)+\frac{r}{\beta}-1\right]\\
		&=\frac{1}{\beta}\ln(1+z^*)+\frac{1}{\beta}\left[n(q^*)+\left(\frac{r}{\beta}+\ln\beta -1\right)-\ell(0)\right] =g(z^*+).
		\end{align*}
		
		We now check the continuity of $zg'(z)$ at $z=z^*$. Let $u^*:=\ln z^*$ and $h^*:=h(u^*)$. Then by construction, $z^*g'(z^*)=h'(u^*)=w(h^*)=q^*$. Meanwhile, $(z^*+)g'(z^*+)=\frac{z^*}{\beta(1+z^*)}=q^*$. This implies the continuity of $g'$ at $z^*$.
		
		Similarly we check the continuity of $z^2g''(z)$ at $z=z^*$. From construction we can deduce
		\begin{align*}
		(z^*)^2g''(z^*)=h''(u^*)-h'(u^*)=w(h^*)[w'(h^*)-1]&=q^*\left(\frac{1}{N'(q^*)}-1\right)
		=-\beta(q^*)^2
		\end{align*}
		where we have used the fact $N'(q)=\frac{1}{\beta}\left(n'(q)+\frac{\beta}{1-\beta q}\right)$ and $n'(q^*)=0$. On the other hand,
		\begin{align*}
		(z^*+)^2g''(z^*+)=-\frac{(z^*)^2}{\beta(1+z^*)^2}=-\beta (q^*)^2.
		\end{align*}
		This completes the proof.
	\end{proof}

When we transform the original HJB equation in Section \ref{sect:heuristics}, a crucial step is a change of the independent variable via $q:=w(h)$ which leads to the transformed value function $n=(n(q))_{0\leq q\leq q^*}$. Proposition \ref{prop:CandidateG} is about the reversal of the transformation. While $q$ is a dummy independent variable associated with the candidate value function $n$ in the transformed system, one should keep in mind that $q$ is related to the original coordinate system through $q:=w[h(u)]=w[h(\ln z)]$. The following lemma provides an important link between the two coordinate systems which will be utilized extensively in many of the subsequent proofs in this paper.

\begin{lemma}

Recall the notations introduced in Proposition \ref{prop:CandidateG}. Write $q:=w[h(u)]=w[h(\ln z)]$. Then $z$ and $q$ are linked via
\begin{align}
z=z(q)=\frac{\beta q}{1-\beta q}\exp\left(-\int_{q}^{q^*}\frac{n'(v)}{\beta v}dv\right).
\label{eq:z_and_q}
\end{align}
In particular, $z:[0,q^*]\to [0,z^*]$ is an strictly increasing function, $z(q)\downarrow 0$ as $q\downarrow 0$ and $z(q)\uparrow z^*$ as $q\uparrow q^*$.
\label{lemma:zqlink}
\end{lemma}

\begin{proof}

Starting from \eqref{eq:defH},
\begin{align*}
\ln z^*-\ln z=\int_{q}^{q^*}\frac{N'(v)}{v}dv&=\int_{q}^{q^*}\left(\frac{n'(v)}{\beta v}+\frac{1}{v(1-\beta v)}\right)dv\\
&=\int_{q}^{q^*}\left(\frac{n'(v)}{\beta v}\right)dv+\ln\frac{q^*}{1-\beta q^*}-\ln\frac{q}{1-\beta q} \\
&=\int_{q}^{q^*}\left(\frac{n'(v)}{\beta v}\right)dv+\ln z^*-\ln\frac{\beta q}{1-\beta q} 
\end{align*}
and we can arrive at \eqref{eq:z_and_q} after a slight rearrangement of the terms. Since $n'(v)>0$ for all $v\in(0,q^*]$, $z=z(q)$ is increasing in $q$.

From the form of \eqref{eq:z_and_q} it is trivial that $z\uparrow z^*$ as $q\uparrow q^*$ because $n'(v)/v$ is bounded near $v=q^*$. To establish $z(q)\downarrow 0$ as $q\downarrow 0$, observe that
\begin{align*}
0\leq z(q)=\frac{\beta q}{1-\beta q}\exp\left(-\int_{q}^{q^*}\frac{n'(v)}{\beta v}dv\right)\leq \frac{\beta q}{1-\beta q}.
\end{align*}
Taking limit gives the desired result.

\end{proof}

Now we formally define the candidate value function in the non-degenerate case and provide a few useful properties.
	
	\begin{prop}
		For $\mu\neq \rho$ or $\mu=\rho>r+\lambda$, define
		\begin{align}
		V^{C}(s,x)=\frac{1}{\beta}\ln x+g^C\left(\frac{s}{x}\right) ,\qquad s>0,x>0
		\label{eq:CandidateValFun}
		\end{align}
		where $g^C$ is defined in Proposition \ref{prop:CandidateG}. Then:
		\begin{enumerate}
		\item $V^C$ can be extended to $s=0$ and $x=0$ by continuity leading to
			\begin{align*}
			V^C(0,x)&=\frac{1}{\beta}\ln x+\frac{1}{\beta}\left[\frac{r}{\beta}+\ln\beta-1\right],&x>0,\\
			V^C(s,0)&=\frac{1}{\beta}\ln s+\frac{1}{\beta}\left[m(q^*)+\left(\frac{r}{\beta}+\ln\beta -1\right)-\ell(0)\right],&s>0.
			\end{align*}
	\item  $V^C(s,x)$ is a concave function and is increasing in both $s$ and $x$.
	
	\item On $\{(s,x):s>xz^*,s\geq0,x\geq 0 \}$, $\mathcal{M}V^C=0$ and $\mathcal{L}V^C\leq 0$.
	\item On $\{(s,x):0\leq s\leq xz^*,x\geq 0,sx\neq 0 \}$, $\mathcal{L}V^C=0$; On $\{(s,x):0<s\leq xz^*,x\geq 0 \}$, $\mathcal{M}V^C\leq 0$.

	\end{enumerate}

	\label{prop:CandidateVProp}
	\end{prop}
	\begin{proof}
	Again we will suppress the superscript $C$ in $g^C$ throughout the proof for brevity.
	
		\begin{enumerate}
		
\item Recall that $u\to h(u)$ is a bijection from $(-\infty,\ln z^*]$ to $(N(0),N(q^*)]$. Then
\begin{align*}
\lim_{z\downarrow 0}g(z)&=\lim_{u\downarrow -\infty}h(u)+\frac{1}{\beta+\lambda}\left[\frac{\lambda}{\beta}\left(\frac{r}{\beta}+\ln\beta-1\right)+\frac{r}{\beta}-1\right] \\
&=N(0)+\frac{1}{\beta+\lambda}\left[\frac{\lambda}{\beta}\left(\frac{r}{\beta}+\ln\beta-1\right)+\frac{r}{\beta}-1\right] \\
&=\frac{1}{\beta}\left[\ell(0)-\ln\frac{1}{\beta}\right]+\frac{1}{\beta+\lambda}\left[\frac{\lambda}{\beta}\left(\frac{r}{\beta}+\ln\beta-1\right)+\frac{r}{\beta}-1\right]
=\frac{1}{\beta}\left[\frac{r}{\beta}+\ln\beta-1\right].
\end{align*}
Thus
\begin{align*}
V^C(0,x):=\lim_{s\downarrow 0}V^C(s,x)=\frac{1}{\beta}\ln x +\lim_{z\downarrow 0}g(z) =\frac{1}{\beta}\ln x+\frac{1}{\beta}\left[\frac{r}{\beta}+\ln\beta-1\right].
\end{align*}

On the other hand, for all $x\neq 0$ and $\frac{s}{x}=z> z^*$ we have
			\begin{align*}
			V^C(s,x)&=\frac{1}{\beta}\ln x+\frac{1}{\beta}\ln(1+s/x)+\frac{1}{\beta}\left[n(q^*)+\left(\frac{r}{\beta}+\ln\beta -1\right)-\ell(0)\right] \\
			&=\frac{1}{\beta}\ln(x+s)+\frac{1}{\beta}\left[n(q^*)+\left(\frac{r}{\beta}+\ln\beta -1\right)-\ell(0)\right]
			\end{align*}
			and hence 
			\begin{align*}
			V^C(s,0):=\lim_{x\downarrow 0}V^C(s,x)=\frac{1}{\beta}\ln s+\frac{1}{\beta}\left[n(q^*)+\left(\frac{r}{\beta}+\ln\beta -1\right)-\ell(0)\right]
			\end{align*}
			since $z^*<\infty$.
			
			\item On $s>xz^*$,
			\begin{align*}
			V^C(s,x)&=\frac{1}{\beta}\ln x+\frac{1}{\beta}\ln(1+s/x)+\frac{1}{\beta}\left[n(q^*)+\left(\frac{r}{\beta}+\ln\beta -1\right)-\ell(0)\right]  \\
			&=\frac{1}{\beta}\ln(s+x)+\frac{1}{\beta}\left[n(q^*)+\left(\frac{r}{\beta}+\ln\beta -1\right)-\ell(0)\right]
			\end{align*}
			which is obviously a concave function increasing in both $s$ and $x$. On $s\leq x z^*$ or equivalently $q\leq q^*$,			
			\begin{align*}
			V^C_x=\frac{1}{x}\left(\frac{1}{\beta}-zg'(z)\right) =\frac{1}{x}\left(\frac{1}{\beta}-h'(u)\right)
			=\frac{1}{x}\left(\frac{1}{\beta}-w(h)\right)
			\geq\frac{1}{x}\left(\frac{1}{\beta}-q^*\right)> 0
			\end{align*}
			as $q^*<1/\beta$ from Proposition \ref{prop:theODE}, and
			\begin{align*}
			V^C_s=\frac{1}{x}g'(z) =\frac{1}{s}zg'(z)
			=\frac{1}{s}h'(u)
			=\frac{1}{s}w(h)> 0.
			\end{align*}
			Since $g$ is second-order smooth at $z=z^*$ by Proposition \ref{prop:CandidateG}, to show that $V^C$ concave it is sufficient to check that the Hessian matrix
			\begin{align*}
			H:=\begin{pmatrix} 
			V^C_{ss} & V^C_{sx} \\
			V^C_{xs} & V^C_{xx} 
			\end{pmatrix}
			\end{align*}
			is semi-negative definite on $s\leq xz^*$. From the transformation adopted,
			\begin{align*}
			V_{ss} &=\frac{g''(z)}{x^2}=\frac{1}{s^2}z^2g''(z)=\frac{1}{s^2}[h''-h']=\frac{w(h)}{s^2}[w'(h)-1]=\frac{q}{s^2}\left[1/N'(q)-1\right]<0
			\end{align*}
			as
			\begin{align}
			N'(q)=\frac{n'(q)}{\beta}+\frac{1}{1-\beta q}\geq\frac{1}{1-\beta q}>1
			\label{eq:NDerRange}
			\end{align}
			given $n$ is increasing. Meanwhile, the determinant of $H$ can be evaluated as
			\begin{align*}
			V_{ss}^C V_{xx}^C-[V^{C}_{xs}]^2 &=\frac{1}{x^2}\left[g''(z)\left(z^2 g''(z)+2z g'(z)-\frac{1}{\beta}\right)-(g'(z)+zg''(z))^2\right]\\ 
			&=-\frac{1}{x^4}\left[\frac{g''(z)}{\beta}+(g'(z))^2\right] \\
			&=-\frac{1}{z^2x^4}\left[\frac{z^2g''(z)}{\beta}+(zg'(z))^2\right] 
			=-\frac{1}{z^2x^4}\left[\frac{h''-h'}{\beta}+(h')^2\right] \\
			&=-\frac{w(h)}{z^2x^4}\left[\frac{w'(h)-1}{\beta}+w(h)\right] =-\frac{q}{z^2x^4}\left[\frac{1/N'(q)-1}{\beta}+q\right].
			\end{align*}
			Hence $\det(H)\geq 0$ if and only if $\frac{1/N'(q)-1}{\beta}+q\leq 0$ which is equivalent to $N'(q)\geq\frac{1}{1-\beta q}$. But again the latter holds due to \eqref{eq:NDerRange}. Thus $V^C$ is concave.

	\item From construction of $V^C$ on $s> xz^*$ it is trivial that and $V^C_s=V^C_x=\frac{1}{\beta(s+x)}$ and hence $\mathcal{M}V^C=0$. Then it remains to show $\mathcal{L}V^C\leq 0$. Suppose $x\neq 0$. A direct evaluation of $\mathcal{L}V^C$ on $\frac{s}{x}=z > z^*$ gives
	\begin{align*}
	\mathcal{L}V^C&=-\ln\frac{1}{\beta x(1+z)}-1+\frac{r}{\beta(1+z)}+\frac{\rho z}{\beta(1+z)}+\frac{(\mu-\rho)^2}{2\sigma^2 \beta} \\
	&\qquad -(\beta+\lambda)\left[\frac{1}{\beta}\ln x+\frac{1}{\beta}\ln(1+z)+\frac{1}{\beta}\left[n(q^*)+\left(\frac{r}{\beta}+\ln\beta -1\right)-\ell(0)\right]\right] \\
	&\qquad +\frac{\lambda}{\beta}\ln x+\frac{\lambda}{\beta}\left[\frac{r}{\beta}+\ln \beta -1\right] \\
	&=(\rho-r)\left[\frac{1}{\beta}\frac{z}{1+z}\right]+\frac{\lambda}{\beta}\ln\left[\frac{1}{\beta}-\frac{1}{\beta}\frac{z}{1+z}\right]+\frac{(\mu-\rho)^2}{2\sigma^2\beta}-\frac{\beta+\lambda}{\beta}n(q^*) \\
	&=\frac{\beta+\lambda}{\beta}\left[m\left(\frac{1}{\beta}\frac{z}{1+z}\right)-n(q^*)\right]  \\
	&\leq \frac{\beta+\lambda}{\beta}\left[m\left(\frac{1}{\beta}\frac{z^*}{1+z^*}\right)-n(q^*)\right]
	=\frac{\beta+\lambda}{\beta}\left[m(q^*)-n(q^*)\right]=0
	\end{align*}
	since $m(q)$ is decreasing on $q\geq q^*=\frac{z^*}{\beta(1+z^*)}$ and $n(q^*)=m(q^*)$ by the definition of $q^*$. The inequality can be extended to $x=0$ by continuity on observing that $m\left(\frac{1}{\beta}\frac{z}{1+z}\right)=m\left(\frac{1}{\beta}\frac{s}{s+x}\right)$.
	
	\item On $z\leq z^*$ the candidate value function $V^C$ is constructed from $n=(n(q))_{0\leq q\leq q^*}$ which by definition solves $\mathcal{L}V^C=0$. Thus we only have to verify that $\mathcal{M}V^C=V^C_x-V^C_s\leq 0$ on $0<z\leq z^*$. Under the transformation adopted the desired inequality is $\frac{1}{x}\left(\frac{1}{\beta}-(1+z)g'(z)\right)\leq 0$ which is equivalent to $zg'(z)\geq  \frac{z}{\beta(1+z)}$ and in turn $q=w(h)=h'=zg'(z)\geq \frac{z}{\beta(1+z)}$ or equivalently $z\leq \frac{\beta q}{1-\beta q}$. But this immediately follows from \eqref{eq:z_and_q}.
	
	
	
	\end{enumerate}
	\end{proof}

\begin{remark}
While we can extend the definition of $\mathcal{L}V^C$ to $x=0$ by continuity, we do not require $\mathcal{M}V^C$ at $s=0$. The rationale is that along $s=0$ the net equity value of the firm is zero and hence no dividend can be paid out, i.e. $d\Phi_{t}=0$ is the only admissible strategy whenever $S_t=0$. The marginal utility contributed by the dividend term $\mathcal{M}V^C d\Phi_t$ is thus zero.
\end{remark}

The following lemma provides some useful results which will facilitate the proof of the verification theorem.
\begin{lemma}
\begin{enumerate}
\item For $V^C$ defined in Proposition \ref{prop:ValFunPropDegen} or \ref{prop:CandidateVProp}, $sV^C_s$ and $\frac{1}{xV_x^C}$ are bounded everywhere.
\item Suppose $\mu\neq \rho$ or $\mu=\rho>\lambda+r$ such that $z^*>0$. Then $\frac{V^C_s}{sV^C_{ss}}$ and $\frac{(V^C_s)^2}{V^C_{ss}}$ are bounded on $0\leq s\leq xz^*$.
\end{enumerate}
\label{lemma:bound}
\end{lemma}

\begin{proof}
\begin{enumerate}
	
\item In the case of $\mu=\rho\leq r+\lambda$ we have a closed-form expression of $V^C$ as in Proposition \ref{prop:ValFunPropDegen} where the desired results can be established easily. 

For the more general case where $V^C$ is defined in Proposition \ref{prop:CandidateVProp}, on $0\leq s \leq xz^*$ we have
\begin{align*}
sV^C_s=zg'(z)=h'(u)=w(h)\in [0,q^*]
\end{align*}
and
\begin{align}
\frac{1}{xV_{x}^C}=\frac{1}{1/\beta-zg'(z)}=\frac{1}{1/\beta-h'(u)}=\frac{1}{1/\beta-w(h)}
\label{eq:consumbound}
\end{align}
which is bounded as $w(h)=q\in [0,q^*]\subset [0,1/\beta)$. Meanwhile, on $s>xz^*$ $V^C$ equals $\frac{1}{\beta}\ln(s+x)$ plus a constant such that $sV^C_s$ and $\frac{1}{xV_x^C}$ are trivially bounded.

\item Similarly,
\begin{align}
\frac{V_s^C}{sV_{ss}^C}=\frac{zg'(z)}{z^2g''(z)}=\frac{h'(u)}{h''(u)-h'(u)}=\frac{1}{w'(h)-1}=\frac{1}{1/N'(q)-1}
\label{eq:pibound}
\end{align}
which is bounded on $0<q\leq q^*$ as $N'(q)$ is continuous, $N'(0)=\frac{n'(0)}{\beta}+1$ and $n'(0)$ is non-zero from Proposition \ref{prop:theODE}. Finally,
\begin{align}
\frac{[V_s^C]^2}{V_{ss}^C}=\frac{[zg'(z)]^2}{z^2g''(z)}=\frac{[h'(u)]^2}{h''(u)-h'(u)}=\frac{w(h)}{w'(h)-1}=\frac{q}{1/N'(q)-1}
\label{eq:integrandbound}
\end{align}
which is also bounded on $0<q\leq q^*$.
\end{enumerate}
\end{proof}

We are now ready to prove Theorem \ref{thm:main}.
	\begin{proof}[Proof of Theorem \ref{thm:main}]
		
		To show that $V^C$ is indeed the value function, it is sufficient to show that $V^C$ is simultaneously an upper bound and a lower bound of $V$ defined in \eqref{eq:OptProbInfinite}.
		
		\begin{enumerate}
		
		\item In this case the candidate value function is given by \eqref{eq:ValFunDegen}. We first show that $V\leq V^C$ which relies on a perturbation argument based on \cite{davis-norman90}. Fix $\epsilon >0$ and define $\tilde{V}^C(s,x):=V^C(s,x+\epsilon)$ such that $\tilde{V}^C$ is bounded below by $V^{C}(0,\epsilon)=\frac{1}{\beta}\ln\epsilon + \frac{1}{\beta}\left[\frac{r}{\beta}+\ln\beta-1\right]$. For an arbitrary admissible strategy $(c,\pi,\Phi)$, let 
		\begin{align*}
			\tilde{M}_t:=\int_{0}^{t}\left[e^{-(\lambda+\beta)u}\ln c_u+\lambda F(X_u)\right] du +e^{-(\lambda+\beta)t} \tilde{V}^C(S_t,X_t). 
		\end{align*}
		Since $\tilde{V}^C$ is $C^{2\times 1}$, generalized Ito's lemma gives
		\begin{align}
			\tilde{M}_t&=\tilde{M}_0+
			\int_{0}^{t}e^{-(\beta+\lambda)u}\Biggl\{\ln c_u-\tilde{V}^C_x c_u+r \tilde{V}^C_x X_u+\rho \tilde{V}^C_s S_u+(\mu-\rho)\tilde{V}^C_s S_u\pi_u\nonumber \\
			&\qquad+\frac{\sigma^2}{2}\tilde{V}^C_{ss} S_u^2\pi_u^2-(\beta+\lambda)\tilde{V}^C+\lambda F(X_u)\Biggr\}du 
			+\int_{0}^{t}e^{-(\beta+\lambda)u}(\tilde{V}^C_s-\tilde{V}^C_x) d\Phi_u \nonumber\\
			&\qquad +\sum_{\upsilon\leq t}e^{-(\beta+\lambda)\upsilon}\left[\tilde{V}^C(S_\upsilon,X_\upsilon)-\tilde{V}^C(S_{\upsilon-},X_{\upsilon-})-\tilde{V}^C_s\Delta S_\upsilon-\tilde{V}^C_x\Delta X_\upsilon\right] \nonumber\\
			&\qquad+\int_{0}^{t}e^{-(\beta+\lambda)u}\sigma\pi_u S_u \tilde{V}_{s}^C dB_u \nonumber \\
			&=:\tilde{M}_0+N^1_t+N^2_t+N^3_t+N^4_t.
			\label{eq:VeriIto} 
		\end{align}
		
		Consider a sequence of stopping times $T_n:=\inf\{t>0:\int_0^t(\pi_u S_u \tilde{V}_{s}^C)^2 du\geq n\}$ under which the stopped local martingale $N^4_{t\wedge T_n}=\int_{0}^{t\wedge T_n}e^{-(\beta+\lambda)u}\sigma\pi_u S_u \tilde{V}_{s}^C dB_u$ is a true martingale for each $n$. Since $\int_{0}^{t}\pi_{u}^{2}du<\infty$ and $sV^C_s$ (and in turn $s\tilde{V}^C_s$) is bounded as shown in part (1) of Lemma \ref{lemma:bound}, we have $T_n\uparrow \infty$ almost surely. 
		
		Using Proposition \ref{prop:ValFunPropDegen},
		\begin{align*}
			\mathcal{M}\tilde{V}^C(s,x)=V^C_x(s,x+\epsilon)-V^C_s(s,x+\epsilon)=\mathcal{M}V^C(s,x+\epsilon)= 0
		\end{align*}
		and
		\begin{align*}
			\mathcal{L}\tilde{V}^C&=-\ln V^C_x(s,x+\epsilon)-1+r xV^C_x (s,x+\epsilon)+\rho sV^C_s(s,x+\epsilon) 
			-(\beta+\lambda)f+\lambda F(x) \\
			&=\Biggl[-\ln V^C_x(s,x+\epsilon)-1+r (x+\epsilon)V^C_x (s,x+\epsilon)+\rho sV^C_s(s,x+\epsilon) 
			-(\beta+\lambda)f+\lambda F(x+\epsilon)\Biggr] \\ &\qquad-r\epsilon V_x^C(s,x+\epsilon)-\lambda (F(x+\epsilon)-F(x)) \\
			&\leq  -r\epsilon V_x^C(s,x+\epsilon)-\lambda (F(x+\epsilon)-F(x)) \leq 0
		\end{align*}
		where the second last inequality is due to the fact that $V^C$ solves $\mathcal{L}V^C=0$ and the last inequality is due to $V^C$ and $F$ being both increasing in $x$. Thus 
		\begin{align*}
		N^1_t\leq \int_{0}^{t}e^{-(\beta+\lambda)u}\mathcal{L}\tilde{V}^C du\leq 0,\qquad N^2_t=\int_{0}^{t}e^{-(\beta+\lambda)u}\mathcal{M}\tilde{V}^C d\Phi_u=0.
		\end{align*}
		Moreover, the concavity property of $V^C$ is inherited by $\tilde{V}^C$ and as such $N^3_t\leq 0$.
		
		Taking expectation on both side of \eqref{eq:VeriIto} at $t\wedge T_n$ leads to
		\begin{align*}
			\mathbb{E}\left[\int_{0}^{t\wedge T_n}\left[e^{-(\lambda+\beta)u}\ln c_u+\lambda F(X_u)\right] du +e^{-(\lambda+\beta)(t\wedge T_n)} \tilde{V}^C(S_{t\wedge T_n},X_{t\wedge T_n})\right]=\mathbb{E}(\tilde{M}_{t\wedge T_n})\leq \tilde{M}_0=\tilde{V}^C(s,x).
		\end{align*}
		Sending $n\uparrow \infty$, monotone convergence and 
		bounded convergence theorem give respectively
		\begin{align*}
			&\lim_{n\uparrow \infty}\mathbb{E}\left[\int_{0}^{t\wedge T_n}\left[e^{-(\lambda+\beta)u}\ln c_u+\lambda F(X_u)\right] du\right]\\
			&=\lim_{n\uparrow \infty}\mathbb{E}\left[\int_{0}^{t\wedge T_n}\left[e^{-(\lambda+\beta)u}\ln c_u+\lambda F(X_u)\right]^{+} du\right]-\lim_{n\uparrow \infty}\mathbb{E}\left[\int_{0}^{t\wedge T_n}\left[e^{-(\lambda+\beta)u}\ln c_u+\lambda F(X_u)\right]^{-} du\right]\\
			&=\mathbb{E}\left[\int_{0}^{t}\left[e^{-(\lambda+\beta)u}\ln c_u+\lambda F(X_u)\right] du\right]
		\end{align*}
		and
		\begin{align*}
			\lim_{n\uparrow \infty}\mathbb{E}\left[e^{-(\lambda+\beta)(t\wedge T_n)} \tilde{V}^C(S_{t\wedge T_n},X_{t\wedge T_n})\right]&\geq \lim_{n\uparrow \infty}\mathbb{E}\left[e^{-(\lambda+\beta)(t\wedge T_n)} \min(\tilde{V}^C(S_{t\wedge T_n},X_{t\wedge T_n}),0)\right] \\&=\mathbb{E}\left[e^{-(\lambda+\beta)t} \min(\tilde{V}^C(S_{t},X_{t}),0)\right].
		\end{align*} 
		Hence we obtain
		\begin{align*}
			\mathbb{E}\left[\int_{0}^{t}\left[e^{-(\lambda+\beta)u}\ln c_u+\lambda F(X_u)\right] du +e^{-(\lambda+\beta)t} \min(\tilde{V}^C(S_{t},X_{t}),0)\right]\leq \tilde{V}^C(s,x)=V^C(s,x+\epsilon).
		\end{align*}
		On letting $t\uparrow \infty$ and then $\epsilon \downarrow 0$, we can conclude 
		\begin{align*}
			\mathbb{E}\left[\int_{0}^{\infty}\left[e^{-(\lambda+\beta)u}\ln c_u+\lambda F(X_u)\right] du \right]\leq V^C(s,x)
		\end{align*}
		and for any admissible $(c,\pi,\Phi)$ and thus $V(s,x)\leq V^C(s,x)$.
		
		To show that $V^C\leq V$ it is sufficient to demonstrate an admissible strategy which attains the candidate value function. Consider the strategy of liquidating the firm by distributing the entire equity to investors in form of dividends, and then investors consume their private wealth at a rate of $\beta$. In other words, the candidate optimal dividend and consumption policy are given by $\Phi^*_t=s$ and $c^*_t=\beta X^*_t$ respectively for $t\geq 0$. The resulting wealth process $X^*$ is thus the solution to
		\begin{align*}
		dX^*_t=(r-\beta)X^*_tdt,\qquad X^*_0=s+x
		\end{align*}
		and hence $ X_t^*=(s+x)e^{(r-\beta)t}$. The candidate optimal consumption policy can be written as
		\begin{align*}
		c^*_t=\beta X^*_t=\beta (s+x)e^{(r-\beta)t}.
		\end{align*}
		The resulting expected lifetime utility is
		\begin{align*}
		\mathbb{E}\left[\int_0^\infty e^{-\beta t}\ln\left(\beta (s+x)e^{(r-\beta)t}\right) dt \right]=\frac{1}{\beta}\ln(s+x)+\frac{1}{\beta}\left[\frac{r}{\beta}+\ln\beta-1\right]=V^C(s,x).
		\end{align*}
		Thus $V^C\leq V$.
			
		\item The proof of $V\leq V^C$ is omitted since it is almost identical to part (1), except that $N^1_t\leq 0$ and $N^2_t\leq 0$ are now established by part (3) and (4) of Proposition \ref{prop:CandidateVProp}. To show that $V\geq V^C$, we again want to demonstrate there exists an admissible strategy under which $		 V^C(s,x)=\mathbb{E}\left[\int_{0}^\infty e^{-(\beta+\lambda)t}[\ln c_t+\lambda F(X_t)]dt\right]$. Suppose the initial value $(s,x)$ is such that $\frac{s}{x}\leq z^*$. Define the feedback controls $c^*=(c^*_t)_{t\geq 0}$, $\pi^*=(\pi^*_t)_{t\geq 0}$ as in \eqref{eq:FeedbackControls}, and $\Phi^*:=(\Phi^*_t)_{t\geq 0}$ an adapted, local time strategy which keeps $\frac{S_t}{X_t}\leq z^*$. By part (2) of Lemma \ref{lemma:bound}, $\pi^*(s,x)$ and $c^*(s,x)/x$ are bounded and thus $(c^*,\pi^*)$ is a pair of valid consumption/investment policy. Denote by $(S^*,X^*)=(S^*_t,X^*_t)_{t\geq 0}$ the state variable processes evolving under these controls. 
		
		Let $M^*_t:=\int_{0}^{t}\left[e^{-(\lambda+\beta)u}\ln c^*_u+\lambda F(X^*_u)\right] du +e^{-(\lambda+\beta)t} V^{C}(S^*_t,X^*_t)$ be the value process under $(c^*,\pi^*,\Phi^*)$. Using part (4) of Proposition \ref{prop:CandidateVProp} and the fact that $d\Phi^*_t=0$ on $Z^*_t:=S^*_t/X^*_t\leq z^*$, Ito's lemma gives
		$M_t^*=M_0^*+\int_0^t \sigma \pi_u^* S_u^* V^C_s dB_u$. By part (2) of Lemma \ref{lemma:bound} the integrand of the stochastic integral is bounded and thus it is a true martingale such that
		\begin{align}
		\mathbb{E}\left[\int_{0}^{t}e^{-(\lambda+\beta)u}\left(\ln c^*_u + \lambda F(X^*_u)\right) du\right]+ \mathbb{E}\left[e^{-(\lambda+\beta)t} V^{C}(S^*_t,X^*_t)\right]=\mathbb{E}(M^*_t)= M^*_0 = V^{C}(s,x).
		\label{eq:VeriMg}
		\end{align}
		
		To show that $(c^*,\pi^*,\Phi^*)$ is admissible, we want to demonstrate that $S^*_t\geq 0$ and $X^*_t\geq 0$ for all $t$ and also $T:=\inf\{t\geq 0: (S^*_t,X^*_t)\in (0,0)\}=\infty$. The design of $\Phi^*$ immediately implies $S^*_t\geq 0$ and $X^*_t\geq 0$. Applying Ito's lemma directly to $V^C(S_t^*,X_t^*)$ gives
		\begin{align}
		V^C(S_t^*,X_t^*)&=V^C(s,x)+\int_{0}^{t}\left[rV_x^C X_u^* + \rho V_s^C S^*_u +(\mu-\rho)V^C_s S_u^* \pi^*_u+\frac{\sigma^2}{2}V_{ss}(\pi_u^*S_u^*)^2-V_x c_u^*\right]du \nonumber \\
		&\qquad +\int_0^t(V^C_x-V^C_s)d\Phi^*_u+\sigma\int_0^t\pi^*_uV_s^CS^*_u dB_u \nonumber \\
		&=V^C(s,x)+\int_0^t\left[-\ln c_u^*+(\beta+\lambda)V^C-\lambda F(X^*_u)\right]du+\sigma\int_0^t\pi^*_uV_s^CS^*_u dB_u \nonumber \\
		&=V^C(s,x)+\int_0^t\left[\ln V^C_x+(\beta+\lambda)V^C-\lambda F(X^*_u)\right]du-\frac{\mu-\rho}{\sigma}\int_0^t\frac{(V^C_s)^2}{V^C_{ss}} dB_u \nonumber \\
		&=:V^C(s,x)+\int_0^t f_1(S^*_u,X^*_u)du+\int_0^t f_2(S^*_u,X^*_u) dB_u 
		\label{eq:ItoVC}
		\end{align}
		where we have used the fact that $\mathcal{L}V^C=0$ on $z\leq z^*$.
		But
		\begin{align*}
		f_1(s,x)&=\ln V_x^C+(\beta+\lambda)V^C-\lambda F(x) \\
		&=\ln\left(\frac{1}{\beta}-zg'(z)\right)+(\beta+\lambda)g(z)-\frac{\lambda}{\beta}\left(\frac{r}{\beta}+\ln \beta -1 \right) \\
		&=\ln\left(\frac{1}{\beta}-h'(u)\right)+(\beta+\lambda)h(u) +\frac{r}{\beta}-1\\
		&=\ln\left(\frac{1}{\beta}-q\right)+(\beta+\lambda)N(q) +\frac{r}{\beta}-1
		=n(q)
		\end{align*}
		which is bounded on $q\leq q^*$ and $f_2(s,x)$ is bounded as well by part (2) of Lemma \ref{lemma:bound}. Now
		\begin{align*}
		V^C(S^*_{t\wedge T},X^*_{t\wedge T})=V^C(s,x)+\int_0^{t\wedge T} f_1(S^*_u,X^*_u)du+\int_0^{t\wedge T} f_2(S^*_u,X^*_u) dB_u
		\end{align*}
		and thus we must have $T=\infty$ because $V^C(S^*_{T},X^*_{T})=V^C(0,0)=-\infty$ but the integrands on the right hand side are bounded.
		
		\eqref{eq:ItoVC} and the boundedness of $f_1$ and $f_2$ also allow us to deduce the transversality condition $\lim_{t\to\infty}\mathbb{E}\left[e^{-(\beta+\lambda)t}V^{C}(S^*_t,X^*_t)\right]=0$. Upon taking limit $t\to\infty$ in \eqref{eq:VeriMg} we can conclude $V(s,x)\geq \mathbb{E}\left[\int_{0}^{\infty}e^{-(\lambda+\beta)u}\left(\ln c^*_u + \lambda F(X^*_u)\right) du\right] = V^C(s,x)$ on $\frac{s}{x}\leq z^*$.
		
		Finally, if the initial value $(s,x)$ is such that $\frac{s}{x}>z^*$, then consider a strategy of paying a discrete dividend $D^*=\frac{s-z^*x}{1+z^*}$ at time zero such that the ex-dividend equity to private wealth ratio is restored to $z^*$, and then follow the candidate optimal strategies $(c^*,\pi^*,\Phi^*)$ described in the regime of $\frac{s}{x}\leq z^*$ thereafter. By construction of $V^C$ on $z>z^*$, $V^C(s,x)=V^C(s-D^*,x+D^*)$. Then \eqref{eq:VeriMg} gives
		\begin{align*}
		\mathbb{E}\left[\int_{0}^{t}e^{-(\lambda+\beta)u}\left(\ln c^*_u + \lambda F(X^*_u)\right) du\right]+ \mathbb{E}\left[e^{-(\lambda+\beta)t} V^{C}(S^*_t,X^*_t)\right]=V^{C}(s-D^*,x+D^*) = V^{C}(s,x).
		\end{align*}
		and again we can conclude $V(s,x)\geq V^C(s,x)$.
		
	\end{enumerate}
		
	\end{proof}
	
\section{Concluding remarks}
\label{sect:conclude}

We develop a continuous-time stochastic control model which jointly determines the optimal dividend policy and capital structure of a defaultable firm as well as the consumption strategy of its risk averse equity investors. We give a complete characterization of the solution to the problem. The optimal dividend policy is a local time strategy which keeps the ratio of the firm's equity value to investors' wealth below a target threshold. Comparative statics of economic importance are derived where the impact of default risk on the corporate policies is highlighted. A firm facing a higher default risk has stronger incentive to pay out dividends aggressively as a precautionary move to preserve value for investors against potential default. To offset the negative effect on investment due to the shrunk equity basis, the firm adopts a higher leverage level to boost its return. This feature can potentially be interesting to mainstream finance literature as dividends and leverage decision of a firm now reflect its riskiness (default probability), and hence they could have important asset pricing implications.

We have exclusively focused on the equity value evolution without considering the payoff to the bondholders. In particular, the corporate yield $\rho$ is a given constant. A possible variant of the current model may involve bondholders who understand the default probability of the firm and adjust the cost of debt accordingly. An example of the modeling strategy can be found in \cite{lambrecht-tse18}, where risk-neutral bondholders charge a fair corporate yield as a function of the leverage level $\pi$ and default risk $\lambda$. Our analysis can be extended in a similar fashion and this can potentially shed colors on the issues of agency cost of debt and their impacts on the corporate policies although the analysis might then rely more heavily on numerical studies.

Investors have logarithm utility function in the current model. A natural and tempting extension of the model is to consider a more general power utility function such that the effect of risk aversion can be investigated. Unfortunately, it appears difficult to apply the same set of transformation scheme to facilitate the analysis of the HJB system since the ``bequest'' term $\lambda F(x)$ now has a multiplicative (rather than additive) form and the resulting first order system $n'=O(q,n)$ is much more complicated. Moreover, we also expect that under power utility function the issue of well-posedness will lead to extra complications of the problem.\footnote{For utility function in form of $u(c)=c^{1-R}/(1-R)$ where $R\in(0,\infty)\setminus \{1\}$ is the risk aversion level, the parameter combination of $R<1$ and $\beta\leq (1-R)r$ leads to an ill-posed problem since the deterministic optimal consumption problem under such parameters is ill-posed and thus a version of the problem with the defaultable firm is also ill-posed.} Identification of the exact conditions under which well-posedness holds for stochastic control problems of this type has historically been a very difficult task. For example, since the rigorous formulation of the Merton consumption/investment problem under transaction costs by \cite{davis-norman90}, it has taken more than two decades for the precise well-posedness conditions to be established by \cite{choi-sirbu-zitkovic13}. A full generalization of the model in the current paper to power utility function should prove to be a challenging open problem for future research.

In our framework, the only outside investment option available for the investors is the retail saving account. Another possible direction of future research is to allow equity investors to also invest in another risky market asset which can potentially be correlated to the risky asset of the firm. While it is expected that the extra dimension introduced will bring significant challenges to the analysis of the underlying HJB equation, it is perhaps not an impossible task in view of the recent progress by \cite{hobson-tse-zhu16} who completely solve a multi-asset Merton problem with transaction costs (albeit the special case where transaction cost is only payable for one of the assets).

\bibliographystyle{apalike}
\bibliography{ref}

\section{Appendix}

\subsection{A bound of $n$}

The following lemma will be useful when establishing the proofs related to the dependence of the optimal controls on the state variables.

\begin{lemma}
	Suppose $\mu\neq \rho$ and recall the notations introduced in Proposition \ref{prop:theODE}. Let $\chi(q):=n(q)-\ell(q)$. Then $\chi(q)\leq \alpha q$ for $0\leq q\leq q^*$ where $\alpha$ is defined to be the positive root to the quadratic equation \eqref{eq:QuadEq}.
\label{lemma:nbound}
\end{lemma}

\begin{proof}
	On differentiating both side of \eqref{eq:chiODE}, setting $q\to 0$ and applying L'Hopital's rule we can obtain
	\begin{align*}
	\chi''(0+)=-\frac{\beta^2(\mu-\rho)^2}{4(\beta+\lambda)^2\sigma^2\alpha^2}\chi''(0+)-\frac{\beta^3}{\beta+\lambda}
	\end{align*}
	which give
	\begin{align*}
	\chi''(0+)=-\frac{\beta^3/(\beta+\lambda)}{1+\frac{\beta^2(\mu-\rho)^2}{4(\beta+\lambda)^2\sigma^2\alpha^2}}<0.
	\end{align*}
	Hence $\chi$ is concave near $q=0$ and then $\chi$ must be initially lying below $L(q):=\chi'(0)q=\alpha q$.
	
	Write the ODE \eqref{eq:chiODE} as $\chi'(q)=\hat{O}(q,\chi(q))$. Then
	\begin{align*}
	\hat{O}(q,L(q))=\hat{O}(q,\alpha q)&=\frac{\beta^2}{\beta+\lambda}\frac{(\mu-\rho)^2}{2\sigma^2\alpha}-\frac{\beta}{\beta+\lambda}\left(\rho-r+\frac{(\mu-\rho)^2}{2\sigma^2}\right)-\frac{\beta}{1-\beta q}\left(\beta q-\frac{\lambda}{\beta+\lambda}\right) \\
	&=\frac{\beta^2}{\beta+\lambda}\frac{(\mu-\rho)^2}{2\sigma^2\alpha}-\frac{\beta}{\beta+\lambda}\left(\rho-r+\frac{(\mu-\rho)^2}{2\sigma^2}\right)-\frac{1}{\lambda+\beta}\frac{\beta^3 q}{1-\beta q}+\frac{\lambda \beta}{\lambda+\beta} \\
	&\leq \frac{\beta^2}{\beta+\lambda}\frac{(\mu-\rho)^2}{2\sigma^2\alpha}-\frac{\beta}{\beta+\lambda}\left(\rho-r+\frac{(\mu-\rho)^2}{2\sigma^2}\right)+\frac{\lambda \beta}{\lambda+\beta}  \\
	&=\frac{A(0,0)}{\alpha}+B(0) =\alpha = L'(q).
	\end{align*}
	Thus $\chi(q)$ can only downcross $L(q)$ and from this we conclude $\chi(q)\leq L(q)=\alpha q$. 
	
\end{proof}

\subsection{Dependence of the optimal controls on the state variables}
\label{appsect:statevar}

	\begin{proof}[Proof of Proposition \ref{prop:policyrange}]
	\begin{enumerate}
		\item Recall from Proposition \ref{prop:theODE} that $q^*\in\left(\frac{1}{\beta}\frac{(\rho-r-\lambda)^{+}}{(\rho-r-\lambda)^{+}+\lambda},\frac{1}{\beta}\right)$. The result follows immediately under the relationship $z^*=\frac{\beta q^*}{1-\beta q^*}$.
		
		\item The result mainly follows from part (2) of Lemma \ref{lemma:bound} that
		\begin{align}
		\pi^*(s,x)=\frac{\mu-\rho}{\sigma^2}\left[-\frac{V_s}{sV_{ss}}\right]=\frac{\mu-\rho}{\sigma^2}\left[- \frac{zg'(z)}{z^2g''(z)}\right]=:\frac{\mu-\rho}{\sigma^2}\theta(z)
		\label{eq:pi}
		\end{align} 
		such that $\theta(z)=- \frac{zg'(z)}{z^2g''(z)}=\frac{1}{1-1/N'(q)}$ which is positive, strictly larger than one and bounded since $N'$ is a positive function bounded away from $1$ on $q\in(0,q^*]$.
		
		It remains to show that $\theta(z)$ is decreasing in $z$ which is equivalent to showing that $N'(q)$ is increasing in $q$. By definition of $N$, we can obtain $	N''(q)=\frac{n''(q)}{\beta}+\frac{\beta}{(1-\beta q)^2}$. We work out the second derivative of $n$ as
		\begin{align*}
		n''(q)&=\frac{d}{dq}O(q,n(q)) =\frac{d}{dq}\left[\frac{\beta^2 q}{1-\beta q}\frac{m(q)-n(q)}{n(q)-\ell(q)}\right] 
		=\frac{d}{dq}\left\{\frac{\beta^2 q}{1-\beta q}\left[\frac{m(q)-\ell(q)}{n-\ell(q)}-1\right]\right\} \\
		&=\frac{\beta^2}{\beta+\lambda}\frac{(\mu-\rho)^2}{2\sigma^2}\frac{[n(q)-\ell(q)]-q[n'(q)-\ell'(q)]}{[n(q)-\ell(q)]^2}-\frac{\beta^2}{(1-\beta q)^2}
		\end{align*}
		and hence
		\begin{align*}
		N''(q)=\frac{\beta}{\beta+\lambda}\frac{(\mu-\rho)^2}{2\sigma^2}\frac{[n(q)-\ell(q)]-q[n'(q)-\ell'(q)]}{[n(q)-\ell(q)]^2}=\frac{\beta}{\beta+\lambda}\frac{(\mu-\rho)^2}{2\sigma^2}\frac{\chi(q)-q\chi'(q)}{[\chi(q)]^2}
		\end{align*}
		where $\chi(q):=n(q)-\ell(q)$ as introduced in the proof of Proposition \ref{prop:theODE}. Hence to show that $N'$ is increasing it is necessary and sufficient to show that $\chi'(q)\leq \frac{\chi(q)}{q}$ for all $0<q\leq q^*$.
		
		Suppose on contrary that there exists $\bar{q}\in(0,q^*]$ such that $\chi'(\bar{q})> \frac{\chi(\bar{q})}{\bar{q}}=:\bar{\alpha}$. From Lemma \ref{lemma:nbound}, $\chi(q)\leq \alpha q$ for $q\in(0,\bar{q}]$ where $\alpha:=\chi'(0)$ and thus we must have $\bar{\alpha} \leq \alpha$. Let $L(q):=\bar{\alpha} q$. Then we have $L'(\bar{q})=\bar{\alpha} <\chi'(\bar{q})$ and hence $\chi(q)$  upcrosses $L(q)$ at $q=\bar{q}$.
		
		Since $\chi'(0)=\alpha>\bar{\alpha}$, $\chi(q)$ must initially be large than $L(q)$ for $q$ near zero. Hence there must exist some $\tilde{q}<\bar{q}$ where $\chi$ downcrosses $L$ at $q=\tilde{q}$. But, recall the definition of $\hat{O}$ in \eqref{eq:chiODE},
		\begin{align*}
		\bar{\alpha}<\chi'(\bar{q})&=\hat{O}(\bar{q},\chi(\bar{q}))=\hat{O}(\bar{q},\bar{\alpha} \bar{q})\\
		&=\frac{\beta^2}{\beta+\lambda}\frac{(\mu-\rho)^2}{2\sigma^2\bar{\alpha}}-\frac{\beta}{\beta+\lambda}\left(\rho-r+\frac{(\mu-\rho)^2}{2\sigma^2}\right)-\frac{\beta}{1-\beta \bar{q}}\left(\beta \bar{q}-\frac{\lambda}{\beta+\lambda}\right) \\
		&=\frac{\beta^2}{\beta+\lambda}\frac{(\mu-\rho)^2}{2\sigma^2\bar{\alpha}}-\frac{\beta}{\beta+\lambda}\left(\rho-r+\frac{(\mu-\rho)^2}{2\sigma^2}\right)-\frac{1}{\lambda+\beta}\frac{\beta^3 \bar{q}}{1-\beta \bar{q}}+\frac{\lambda \beta}{\lambda+\beta} \\
		&< \frac{\beta^2}{\beta+\lambda}\frac{(\mu-\rho)^2}{2\sigma^2\bar{\alpha}}-\frac{\beta}{\beta+\lambda}\left(\rho-r+\frac{(\mu-\rho)^2}{2\sigma^2}\right)-\frac{1}{\lambda+\beta}\frac{\beta^3 \tilde{q}}{1-\beta \tilde{q}}+\frac{\lambda \beta}{\lambda+\beta} \\
		&=\hat{O}(\tilde{q},\bar{\alpha} \tilde{q})=\hat{O}(\tilde{q},\chi(\tilde{q}))=\chi'(\tilde{q}).
		\end{align*}
		But this contradicts the fact that $\chi$ downcrosses $L$ at $q=\tilde{q}$.
		
		\item The result is immediate from \eqref{eq:consumbound} where
		\begin{align}
		\frac{c^*(s,x)}{x}=\frac{1}{1/\beta-zg'(z)}= \frac{1}{1/\beta-q}.
		\label{eq:consume}
		\end{align} 
		The above is increasing in $q$ and in turn $z$ and it tends to $\beta$ as $q\downarrow 0$ or equivalently $z\downarrow 0$.
	
	\end{enumerate}
\end{proof}

\begin{proof}[Proof of Proposition \ref{prop:PolicySpecialCase}]

$q^*=\frac{\rho-r-\lambda}{\beta(\rho-r)}$ under $\mu=\rho> \lambda+r$ and hence $z^*=\frac{\beta q^*}{1-\beta q^*}=\frac{\rho-r-\lambda}{\lambda}$. $\pi^{*}=0$ is trivial when $\mu=\rho$. Finally, the expression of $c^*(s,x)$ can be established using \eqref{eq:consume} and the fact that $n=m$ when $\mu=\rho>\lambda+r$. Thus \eqref{eq:z_and_q} in Lemma \ref{lemma:zqlink} can be further simplified to
\begin{align}
z=z(q)&=\frac{\beta q}{1-\beta q}\exp\left(-\int_{q}^{\frac{\rho-r-\lambda}{\beta(\rho-r)}}\frac{m'(v)}{\beta v}dv\right) \nonumber\\
&=\frac{\beta q}{1-\beta q}\exp\left[-\int_{q}^{\frac{\rho-r-\lambda}{\beta(\rho-r)}}\frac{1}{\beta +\lambda}\left(\frac{\rho-r}{v}-\frac{\lambda}{v(1-\beta v)}\right)dv\right]
\label{eq:zqSpeicalCaseIntegral}
\end{align}
and we can arrive at \eqref{eq:closeform_zq}.

\end{proof}

\subsection{Comparative statics}
\label{appsect:compstat}

\begin{proof}[Proof of Proposition \ref{prop:compstat}]
	
	In the case of $\mu= \rho>\lambda+r$, $z^*=\frac{\rho-r-\lambda}{\lambda}$ and $\pi^*=0$ and hence their comparative statics are trivial. While it may be less trivial to deduce the comparative statics of $c^*(s,x)$ directly from \eqref{eq:closeform_zq}, one trick is to observe its integral form as in \eqref{eq:zqSpeicalCaseIntegral} to deduce that $z=z(q)$ is increasing in $\lambda$ and $r$. As $z=z(q)$ is an increasing bijection, its inverse function $q=q(z)$ is decreasing in $\lambda$ and $r$. This monotonicity is then inherited by $c^*(s,x)=\frac{x}{1/\beta-q(s/x;\lambda,r)}$.
	
	Now we proceed to give the proof in the general case of $\mu\neq \rho$:
	\begin{enumerate}
		
		\item Recall that the transformed value function $n$ is the solution to the ODE $n'=O(q,n)$ where $O$ is defined in Proposition \ref{prop:theODE}, and the transformed dividend payment boundary is given by $q^*:=\inf\{q>0:n(q)\geq m(q)\}$.
		
		Let $b(q):=\frac{\beta+\lambda}{\beta}(m(q)-n(q))$. Then the ODE becomes $m'(q)-\frac{\beta}{\beta+\lambda}b'(q)=O\left(q,m(q)-\frac{\beta}{\beta+\lambda}b(q)\right)$ which can be written as
		\begin{align}
		b'(q)=\rho-r-\frac{\lambda}{1-\beta q}-\frac{(\beta+\lambda)\beta q}{1-\beta q}\frac{b(q)}{\frac{(\mu-\rho)^2}{2\sigma^2}\left(\frac{1}{\beta}-q\right)-b(q)}=:P(q,b(q))
		\label{eq:FormP}
		\end{align}
		subject to initial condition $b(0)=\frac{\beta+\lambda}{\beta}(m(0)-\ell(0))=\frac{(\mu-\rho)^2}{2\beta\sigma^2}$. The dividend payment boundary can now be expressed as $q^*:=\inf\{q>0:b(q)\leq 0\}$. Note that $P(q,b)$ is decreasing in $\lambda$ for as long as $0\leq \frac{(\mu-\rho)^2}{2\sigma^2}\left(\frac{1}{\beta}-q\right)$ which must be satisfied along the solution trajectory $b=(b(q))_{0\leq q\leq q^*}$ since the transformed value function $n=n(q)$ always lies between $m(q)$ and $\ell(q)$ on $[0,q^*]$. 
		
		Consider $\lambda^{hi}>\lambda^{lo}$ and denote by $b^{hi}$ (resp. $b^{lo}$) the solution to the ODE $b'=P(q,b(q);\lambda^{hi})$ (resp. $b'=P(q,b(q);\lambda^{lo})$) with initial condition $b(0)=\frac{(\mu-\rho)^2}{2\beta\sigma^2}$. Since $P(q,b;\lambda^{hi})<P(q,b;\lambda^{lo})$, $b^{hi}$ can only downcross $b^{lo}$. Thus $b^{hi}$ is dominated by $b^{lo}$ at least up to $\min(q^{*}_{hi},q^{*}_{lo})$, where $q^*_{hi}:=\inf\{q>0:b^{hi}(q)\leq 0\}$ (and $q^*_{lo}$ is defined similarly). Hence we must have $q^{*}_{hi}<q^{*}_{lo}$ from which we conclude $q^*$ and in turn $z^*=\frac{\beta q^*}{1-\beta q^*}$ are both decreasing in $\lambda$.
		
		The exact same argument can be used to establish the comparative statics of $z^*$ with respect to $\mu$, $\sigma$ and $r$. From \eqref{eq:FormP} it is easy to see that $P(q,b)$ is increasing in $\frac{(\mu-\rho)^2}{\sigma^2}$ (while keeping all the other parameters fixed) and is decreasing in $r$. The result follows immediately.
		
		\item From \eqref{eq:pi}, $(\mu-\rho)\pi^*(s,x)\propto\frac{1}{1-1/N'(q)}$ and hence to show that $(\mu-\rho)\pi^*$ is increasing in $\lambda$ it is sufficient to show that $N'(q)=N'(q(z;\lambda); \lambda)$ is decreasing in $\lambda$. Using the substitution of $b(q):=\frac{\beta+\lambda}{\beta}(m(q)-n(q))$ again, we have
		\begin{align*}
		N'(q;\lambda)=\frac{n'(q;\lambda)}{\beta}+\frac{1}{1-\beta q} 
		&=\frac{m'(q;\lambda)-\frac{\beta}{\beta+\lambda}b'(q;\lambda)}{\beta}+\frac{1}{1-\beta q}\\
		&=\frac{\beta q}{1-\beta q}\frac{b(q;\lambda)}{\frac{(\mu-\rho)^2}{2\sigma^2}\left(\frac{1}{\beta}-q\right)-b(q;\lambda)}+\frac{1}{1-\beta q}
		\end{align*}
		and thus $N'(q;\lambda)$ is decreasing in $\lambda$ under a fixed $q$ as $b(q;\lambda)$ is decreasing in $\lambda$ as shown in part (1) of this proof. Then
		\begin{align*}
		\frac{d}{d\lambda} N'(q(\lambda);\lambda)=\frac{\partial}{\partial \lambda} N'(q(\lambda);\lambda)+\frac{\partial}{\partial q} N'(q(\lambda);\lambda)\frac{\partial }{\partial\lambda}q(\lambda)<0
		\end{align*}
		since $q(\lambda)$ is decreasing in $\lambda$ and $N'(q)$ is increasing in $q$ as already shown in  part (1) of this proof and part (2) of the proof of Proposition \ref{prop:policyrange} respectively. Similarly, we can show that $(\mu-\rho)\pi^*$ is increasing in $r$.
		
		\item From \eqref{eq:consume} the optimal consumption rate per unit wealth is given by $\frac{c^*(s,x)}{x}=\frac{1}{1/\beta-q}$. We first show that the expression is decreasing in $\lambda$ which is equivalent to showing that $q=q(z;\lambda)$ is decreasing in $\lambda$. 
		
		Using \eqref{eq:z_and_q} and the substitution of $b(q):=\frac{\beta+\lambda}{\beta}(m(q)-n(q))$ again, we have
		\begin{align*}
		z=z(q;\lambda)&=\frac{\beta q}{1-\beta q}\exp\left[-\int_{q}^{q^*(\lambda)}\left(\frac{\beta}{1-\beta v}\frac{b(v;\lambda)}{\frac{(\mu-\rho)^2}{2\sigma^2}\left(\frac{1}{\beta}-v\right)-b(v;\lambda)}\right)dv\right].
		\end{align*}
		We have shown in part (1) of the proof that both $q^*(\lambda)$ and $b({}\cdot{};\lambda)$ are decreasing in $\lambda$. Hence $z=z(q;\lambda)$ is increasing in $\lambda$. As $z=z(q;\lambda)$ is increasing in $q$, $q=q(z;\lambda)$ is decreasing in $\lambda$. Hence the result follows. Using the exact same argument, it can be shown that $q=q(z;r)$ and in turn $\frac{c^*(s,x)}{x}$ are decreasing in $r$.
		
	\end{enumerate}

\end{proof}

\end{document}